\documentclass[5p]{elsarticle}

\usepackage{hyperref}
\usepackage{url}
\usepackage{booktabs}
\usepackage{amsfonts}
\usepackage{graphicx}
\usepackage{float}
\usepackage{multirow}
\usepackage{algorithm}
\usepackage{algorithmic}
\usepackage{amsmath}
\usepackage{amsthm}
\usepackage{subcaption}
\usepackage{appendix}
\usepackage{xcolor}
\usepackage{xurl}

\biboptions{authoryear}

\DeclareMathOperator*{\argmax}{argmax}
\newtheorem{theorem}{Theorem}

\makeatletter
\def\ps@pprintTitle{%
 \let\@oddhead\@empty
 \let\@evenhead\@empty
 \def\@oddfoot{}%
 \let\@evenfoot\@oddfoot}
\makeatother

\begin{document}

\begin{frontmatter}

\title{An Application of Deep Reinforcement Learning to Algorithmic Trading}
\author[1]{Thibaut Th\'{e}ate}\corref{cor1}
\ead{thibaut.theate@uliege.be}
\author[1]{Damien Ernst}
\ead{dernst@uliege.be}
\address[1]{Montefiore Institute, University of Li\`{e}ge (All\'{e}e de la d\'{e}couverte 10, 4000 Li\`{e}ge, Belgium)}
\cortext[cor1]{Corresponding author.}

\begin{abstract}
    This scientific research paper presents an innovative approach based on deep reinforcement learning (DRL) to solve the algorithmic trading problem of determining the optimal trading position at any point in time during a trading activity in the stock market. It proposes a novel DRL trading policy so as to maximise the resulting Sharpe ratio performance indicator on a broad range of stock markets. Denominated the Trading Deep Q-Network algorithm (TDQN), this new DRL approach is inspired from the popular DQN algorithm and significantly adapted to the specific algorithmic trading problem at hand. The training of the resulting reinforcement learning (RL) agent is entirely based on the generation of artificial trajectories from a limited set of stock market historical data. In order to objectively assess the performance of trading strategies, the research paper also proposes a novel, more rigorous performance assessment methodology. Following this new performance assessment approach, promising results are reported for the TDQN algorithm.
\end{abstract}

\begin{keyword}
Artificial intelligence \sep deep reinforcement learning \sep algorithmic trading \sep trading policy.
\end{keyword}

\end{frontmatter}

\section{Introduction}
\label{SectionIntroduction}

For the past few years, the interest in \textit{artificial intelligence} (AI) has grown at a very fast pace, with numerous research papers published every year. A key element for this growing interest is related to the impressive successes of \textit{deep learning} (DL) techniques which are based on \textit{deep neural networks} (DNN) - mathematical models directly inspired by the human brain structure. These specific techniques are nowadays the state of the art in many applications such as speech recognition, image classification or natural language processing. In parallel to DL, another field of research has recently gained much more attention from the research community: \textit{deep reinforcement learning} (DRL). This family of techniques is concerned with the learning process of an intelligent agent (i) interacting in a sequential manner with an unknown environment (ii) aiming to maximise its cumulative rewards and (iii) using DL techniques to generalise the information acquired from the interaction with the environment. The many recent successes of DRL techniques highlight their ability to solve complex sequential decision-making problems.\\

Nowadays, an emerging industry which is growing extremely fast is the \textit{financial technology} industry, generally referred to by the abbreviation \textit{FinTech}. The objective of FinTech is pretty simple: to extensively take advantage of technology in order to innovate and improve activities in finance. In the coming years, the FinTech industry is expected to revolutionise the way many decision-making problems related to the financial sector are addressed, including the problems related to trading, investment, risk management, portfolio management, fraud detection and financial advising, to cite a few. Such complex decision-making problems are extremely complex to solve as they generally have a sequential nature and are highly stochastic, with an environment partially observable and potentially adversarial. In particular, \textit{algorithmic trading}, which is a key sector of the FinTech industry, presents particularly interesting challenges. Also called quantitative trading, algorithmic trading is the methodology to trade using computers and a specific set of mathematical rules.\\

The main objective of this research paper is to answer the following question: how to design a novel trading policy (algorithm) based on AI techniques that could compete with the popular algorithmic trading strategies widely adopted in practice? To answer this question, this scientific article presents and analyses a novel DRL solution to tackle the algorithmic trading problem of determining the optimal trading position (long or short) at any point in time during a trading activity in the stock market. The algorithmic solution presented in this research paper is inspired by the popular Deep Q-Network (DQN) algorithm, which has been adapted to the particular sequential decision-making problem at hand. The research question to be answered is all the more relevant as the trading environment presents very different characteristics from those which have already been successfully solved by DRL approaches, mainly significant stochasticity and extremely poor observability.\\

The scientific research paper is structured as follows. First of all, a brief review of the scientific literature around the algorithmic trading field and its main AI-based contributions is presented in Section \ref{SectionLiterature}. Afterwards, Section \ref{SectionFormalisation} introduces and rigorously formalises the particular algorithmic trading problem considered. Additionally, this section makes the link with the reinforcement learning (RL) approach. Then, Section \ref{SectionDRLAlgorithm} covers the complete design of the TDQN trading strategy based on DRL concepts. Subsequently, Section \ref{SectionPerformanceAssessment} proposes a novel methodology to objectively assess the performance of trading strategies. Section \ref{SectionResults} is concerned with the presentation and discussion of the results achieved by the TDQN trading strategy. To end this research paper, Section \ref{SectionConclusion} discusses interesting leads as future work and draws meaningful conclusions.

\section{Literature review}
\label{SectionLiterature}

To begin this brief literature review, two facts have to be emphasised. Firstly, it is important to be aware that many sound scientific works in the field of algorithmic trading are not publicly available. As explained in \cite{Li2017}, due to the huge amount of money at stake, private FinTech firms are very unlikely to make their latest research results public. Secondly, it should be acknowledged that making a fair comparison between trading strategies is a challenging task, due to the lack of a common, well-established framework to properly evaluate their performance. Instead, the authors generally define their own framework with their evident bias. Another major problem is related to the trading costs which are variously defined or even omitted.\\

First of all, most of the works in algorithmic trading are techniques developed by mathematicians, economists and traders who do not exploit AI. Typical examples of classical trading strategies are the \textit{trend following} and \textit{mean reversion} strategies, which are covered in detail in \cite{Chan2009}, \cite{Chan2013} and \cite{Narang2009}. Then, the majority of works applying machine learning (ML) techniques in the algorithmic trading field focus on forecasting. If the financial market evolution is known in advance with a reasonable level of confidence, the optimal trading decisions can easily be computed. Following this approach, DL techniques have already been investigated with good results, see e.g. \cite{Arevalo2016} introducing a trading strategy based on a DNN, and especially \cite{Bao2017} using wavelet transforms, stacked autoencoders and long short-term memory (LSTM). Alternatively, several authors have already investigated RL techniques to solve this algorithmic trading problem. For instance, \cite{Moody2001} introduced a recurrent RL algorithm for discovering new investment policies without the need to build forecasting models, and \cite{Dempster2006} used adaptive RL to trade in foreign exchange markets. More recently, a few works investigated DRL techniques in a scientifically sound way to solve this particular algorithmic trading problem. For instance, one can first mention \cite{Deng2017} which introduced the fuzzy recurrent deep neural network structure to obtain a technical-indicator-free trading system taking advantage of fuzzy learning to reduce the time series uncertainty. One can also mention \cite{Carapuco2018} which studied the application of the deep Q-learning algorithm for trading in foreign exchange markets. Finally, there exist a few interesting works studying the application of DRL techniques to algorithmic trading in specific markets, such as in the field of energy, see e.g. the article \cite{Boukas2020}.\\

To finish with this short literature review, a sensitive problem in the scientific literature is the tendency to prioritise the communication of good results or findings, sometimes at the cost of a proper scientific approach with objective criticism. Going even further, \cite{Ioannidis2005} even states that most published research findings in certain sensitive fields are probably false. Such concern appears to be all the more relevant in the field of financial sciences, especially when the subject directly relates to trading activities. Indeed, \cite{Bailey2014} claims that many scientific publications in finance suffer from a lack of a proper scientific approach, instead getting closer to pseudo-mathematics and financial charlatanism than rigorous sciences. Aware of these concerning tendencies, the present research paper intends to deliver an unbiased scientific evaluation of the novel DRL algorithm proposed.

\section{Algorithmic trading problem formalisation}
\label{SectionFormalisation}

In this section, the sequential decision-making algorithmic trading problem studied in this research paper is presented in detail. Moreover, a rigorous formalisation of this particular problem is performed. Additionally, the link with the RL formalism is highlighted.

\subsection{Algorithmic trading}
\label{SectionAlgorithmicTrading}

Algorithmic trading, also called quantitative trading, is a subfield of finance, which can be viewed as the approach of automatically making trading decisions based on a set of mathematical rules computed by a machine. This commonly accepted definition is adopted in this research paper, although other definitions exist in the literature. Indeed, several authors differentiate the trading decisions (quantitative trading) from the actual trading execution (algorithmic trading). For the sake of generality, algorithmic trading and quantitative trading are considered synonyms in this research paper, defining the entire automated trading process. Algorithmic trading has already proven to be very beneficial to markets, the main benefit being the significant improvement in liquidity, as discussed in \cite{Hendershott2011}. For more information about this specific field, please refer to \cite{Treleaven2013} and \cite{Nuti2011}.\\

There are many different markets suitable to apply algorithmic trading strategies. Stocks and shares can be traded in the stock markets, FOREX trading is concerned with foreign currencies, or a trader could invest in commodity futures, to only cite a few. The recent rise of cryptocurrencies, such as the Bitcoin, offers new interesting possibilities as well. Ideally, the DRL algorithms  developed in this research paper should be applicable to multiple markets. However, the focus will be set on stock markets for now, with an extension to various other markets planned in the future.\\

In fact, a trading activity can be viewed as the management of a portfolio, which is a set of assets including diverse stocks, bonds, commodities, currencies, etc. In the scope of this research paper, the portfolio considered consists of one single stock together with the agent cash. The portfolio value $v_t$ is then composed of the trading agent cash value $v^c_t$ and the share value $v^s_t$, which continuously evolves over time $t$. Buying and selling operations are simply cash and share exchanges. The trading agent interacts with the stock market through an order book, which contains the entire set of buying orders (\textit{bids}) and selling orders (\textit{asks}). An example of a simple order book is depicted in Table \ref{orderBook}. An order represents the willingness of a market participant to trade and is composed of a price $p$, a quantity $q$ and a side $s$ (bid or ask). For a trade to occur, a match between bid and ask orders is required, an event which can only happen if $p_{max}^{bid} \geq p_{min}^{ask}$, with $p_{max}^{bid}$ ($p_{min}^{ask}$) being the maximum (minimum) price of a bid (ask) order. Then, a trading agent faces a very difficult task in order to generate profit: what, when, how, at which price and which quantity to trade. This is the algorithmic trading complex sequential decision-making problem studied in this scientific research paper.

\begin{table}[H]
  \small
  \caption{Example of a simple order book}
  \label{orderBook}
  \centering
  \begin{tabular}{ccc}
    \toprule
    \textbf{Side} $s$ & \textbf{Quantity} $q$ & \textbf{Price} $p$ \\
    \cmidrule(r){1-3}
    Ask & 3000 & 107 \\
    Ask & 1500 & 106 \\
    Ask & 500 & 105 \\
    \hline
    \hline
    Bid & 1000 & 95 \\
    Bid & 2000 & 94 \\
    Bid & 4000 & 93 \\
    \bottomrule
  \end{tabular}
\end{table}

\subsection{Timeline discretisation}
\label{SectionDiscretisation}

Since trading decisions can be issued at any time, the trading activity is a continuous process. In order to study the algorithmic trading problem described in this research paper, a discretisation operation of the continuous timeline is performed. The trading timeline is discretized into a high number of discrete trading time steps $t$ of constant duration $\Delta t$. In this research paper, for the sake of clarity, the increment (decrement) operations $t+1$ ($t-1$) are used to model the discrete transition from time step $t$ to time step $t + \Delta t$ ($t - \Delta t$).\\

The duration $\Delta t$ is closely linked to the trading frequency targeted by the trading agent (very high trading frequency, intraday, daily, monthly, etc.). Such discretisation operation inevitably imposes a constraint with respect to this trading frequency. Indeed, because the duration $\Delta t$ between two time steps cannot be chosen as small as possible due to technical constraints, the maximum trading frequency achievable, equal to $1/\Delta t$, is limited. In the scope of this research paper, this constraint is met as the trading frequency targeted is daily, meaning that the trading agent makes a new decision once every day.

\subsection{Trading strategy}
\label{SectionTradingStrategy}

The algorithmic trading approach is rule based, meaning that the trading decisions are made according to a set of rules: a \textit{trading strategy}. In technical terms, a trading strategy can be viewed as a programmed policy $\pi(a_t|i_t)$, either deterministic or stochastic, which outputs a trading action $a_t$ according to the information available to the trading agent $i_t$ at time step $t$. Additionally, a key characteristic of a trading strategy is its sequential aspect, as illustrated in Figure \ref{tradingStrategy}. An agent executing its trading strategy sequentially applies the following steps:

\begin{enumerate}
    \item Update of the available market information $i_t$.
    \item Execution of the policy $\pi(a_t|i_t)$ to get action $a_t$.
    \item Application of the designated trading action $a_t$.
    \item Next time step $t \rightarrow t+1$, loop back to step 1.
\end{enumerate}

\begin{figure}[H]
    \centering
    \includegraphics[scale=0.33]{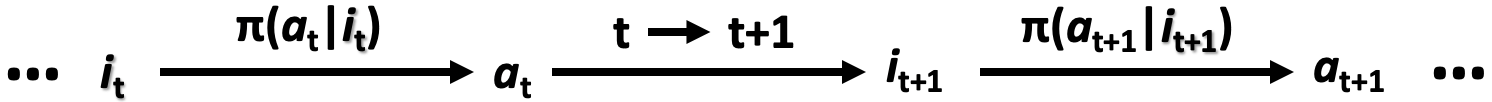}
    \caption{Illustration of a trading strategy execution}
    \label{tradingStrategy}
\end{figure}

In the following subsection, the algorithmic trading sequential decision-making problem, which shares similarities with other problems successfully tackled by the RL community, is casted as an RL problem.

\subsection{Reinforcement learning problem formalisation}
\label{SectionRLFormalism}

As illustrated in Figure \ref{RLScheme}, reinforcement learning is concerned with the sequential interaction of an agent with its environment. At each time step $t$, the RL agent firstly observes the RL environment of internal state $s_t$, and retrieves an observation $o_t$. It then executes the action $a_t$ resulting from its RL policy $\pi(a_t|h_t)$ where $h_t$ is the RL agent history and receives a reward $r_t$ as a consequence of its action. In this RL context, the agent history can be expressed as $h_t = \{(o_\tau, a_\tau, r_\tau) | \tau = 0, 1, ..., t\}$.\\

Reinforcement learning techniques are concerned with the design of policies $\pi$ maximising an optimality criterion, which directly depends on the immediate rewards $r_t$ observed over a certain time horizon. The most popular optimality criterion is the expected discounted sum of rewards over an infinite time horizon. Mathematically, the resulting optimal policy $\pi^*$ is expressed as the following:

\begin{equation}
    \label{EquationOptimalPolicy}
    \pi^* = \argmax_{\pi} \mathbb{E}[R|\pi]
\end{equation}

\begin{equation}
    \label{EquationReturnRL}
    R = \sum_{t=0}^{\infty} \gamma^t r_t
\end{equation}

The parameter $\gamma$ is the discount factor ($\gamma \in [0,1]$). It determines the importance of future rewards. For instance, if $\gamma = 0$, the RL agent is said to be myopic as it only considers the current reward and totally discards the future rewards. When the discount factor increases, the RL agent tends to become more long-term oriented. In the extreme case where $\gamma = 1$, the RL agent considers each reward equally. This key parameter should be tuned according to the desired behaviour.

\begin{figure}[H]
    \centering
    \includegraphics[scale=0.375]{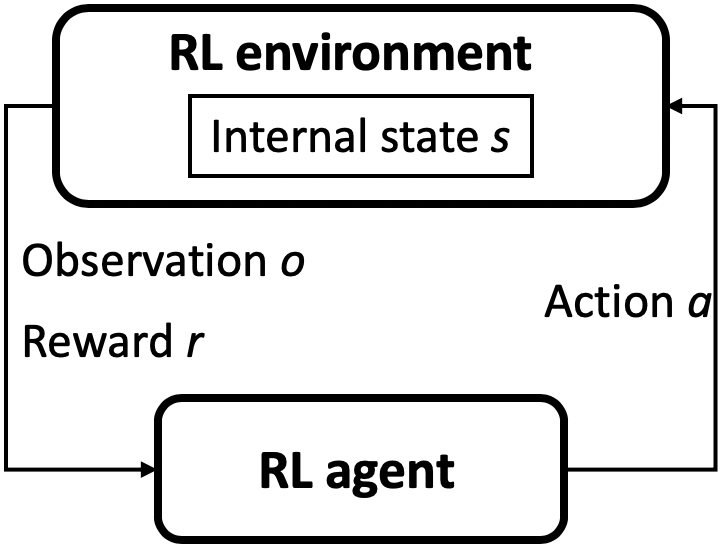}
    \caption{Reinforcement learning core building blocks}
    \label{RLScheme}
\end{figure}

\subsubsection{RL observations}
\label{SectionRLObservations}

In the scope of this algorithmic trading problem, the RL environment is the entire complex trading world gravitating around the RL agent. In fact, this trading environment can be viewed as an abstraction including the trading mechanisms together with every single piece of information capable of having an effect on the trading activity of the agent. A major challenge of the algorithmic trading problem is the extremely poor observability of this environment. Indeed, a significant amount of information is simply hidden to the trading agent, ranging from some companies' confidential information to the other market participants' strategies. In fact, the information available to the RL agent is extremely limited compared to the complexity of the environment. Moreover, this information can take various forms, both quantitative and qualitative. Correctly processing such information and re-expressing it using relevant quantitative figures while minimising the subjective bias is capital. Finally, there are significant time correlation complexities to deal with. Therefore, the information retrieved by the RL agent at each time step should be considered sequentially as a series of information rather than individually.\\

\newpage

At each trading time step $t$, the RL agent observes the stock market whose internal state is $s_t \in \mathcal{S}$. The limited information collected by the agent on this complex trading environment is denoted by $o_t \in \mathcal{O}$. Ideally, this observation space $\mathcal{O}$ should encompass all the information capable of influencing the market prices. Because of the sequential aspect of the algorithmic trading problem, an observation $o_t$ has to be considered as a sequence of both the information gathered during the previous $\tau$ time steps (history) and the newly available information at time step $t$. In this research paper, the RL agent observations can be mathematically expressed as the following:

\begin{equation}
    \label{EquationObservations}
    o_t = \{S(t'),\ D(t'),\ T(t'),\ I(t'),\ M(t'),\ N(t'),\ E(t')\}_{t' = t - \tau}^{t}
\end{equation}

\noindent where:

\begin{itemize}
    \item [$\bullet$] $S(t)$ represents the state information of the RL agent at time step $t$ (current trading position, number of shares owned by the agent, available cash).
    
    \item [$\bullet$] $D(t)$ is the information gathered by the agent at time step $t$ concerning the OHLCV (Open-High-Low-Close-Volume) data characterising the stock market. More precisely, $D(t)$ can be expressed as follows:
    \begin{equation}
    \label{EquationOHLCV}
        D(t) = \{p^O_t,\ p^H_t,\ p^L_t,\ p^C_t,\ V_t\}
    \end{equation}
    \noindent where:
    \begin{itemize}
        \item $p^O_t$ is the stock market price at the opening of the time period $[t - \Delta t,\ t[$.
        \item $p^H_t$ is the highest stock market price over the time period $[t - \Delta t,\ t[$.
        \item $p^L_t$ is the lowest stock market price over the time period $[t - \Delta t,\ t[$.
        \item $p^C_t$ is the stock market price at the closing of the time period $[t - \Delta t,\ t[$.
        \item $V_t$ is the total volume of shares exchanged over the time period $[t - \Delta t,\ t[$.
    \end{itemize}
    
    \item [$\bullet$] $T(t)$ is the agent information regarding the trading time step $t$ (date, weekday, time).
    
    \item [$\bullet$] $I(t)$ is the agent information regarding multiple technical indicators about the stock market targeted at time step $t$. There exist many technical indicators providing extra insights about diverse financial phenomena, such as moving average convergence divergence (MACD), relative strength index (RSI) or average directional index (ADX), to only cite a few.
    
    \item [$\bullet$] $M(t)$ gathers the macroeconomic information at the disposal of the agent at time step $t$. There are many interesting macroeconomic indicators which could potentially be useful to forecast markets' evolution, such as the interest rate or the exchange rate.
    
    \item [$\bullet$] $N(t)$ represents the news information gathered by the agent at time step $t$. These news data can be extracted from various sources such as social media (Twitter, Facebook, LinkedIn), the newspapers, specific journals, etc. Complex sentiment analysis models could then be built to extract meaningful quantitative figures (quantity, sentiment polarity and subjectivity, etc.) from the news. The benefits of such information has already been demonstrated by several authors, see e.g. \cite{Leinweber2011}, \cite{Bollen2011} and \cite{Nuij2014}.
    
    \item [$\bullet$] $E(t)$ is any extra useful information at the disposal of the trading agent at time step $t$, such as other market participants trading strategies, companies' confidential information, similar stock market behaviours, rumours, experts' advice, etc.\\
\end{itemize}

\underline{\textbf{Observation space reduction:}}\\
\indent In the scope of this research paper, it is assumed that the only information considered by the RL agent is the classical OHLCV data $D(t)$ together with the state information $S(t)$. Especially, the reduced observation space $\mathcal{O}$ encompasses the current trading position together with a series of the previous $\tau + 1$ daily open-high-low-close prices and daily traded volume. With such an assumption, the reduced RL observation $o_t$ can be expressed as the following:

\begin{equation}
    \label{EquationObservationsRL}
    o_t = \left\{\{p^O_{t'},\ p^H_{t'},\ p^L_{t'},\ p^C_{t'},\ V_{t'}\}_{t' = t - \tau}^{t},\ P_t\right\}
\end{equation}

\vspace{0.2cm}

\noindent with $P_t$ being the trading position of the RL agent at time step $t$ (either \textit{long} or \textit{short}, as explained in the next subsection of this research paper).

\subsubsection{RL actions}
\label{SectionRLActions}

At each time step $t$, the RL agent executes a trading action $a_t \in \mathcal{A}$ resulting from its policy $\pi(a_t|h_t)$. In fact, the trading agent has to answer several questions: whether, how and how much to trade? Such decisions can be modelled by the quantity of shares bought by the trading agent at time step $t$, represented by $Q_t \in \mathbb{Z}$. Therefore, the RL actions can be expressed as the following:

\begin{equation}
    \label{EquationActions}
    a_t = Q_t
\end{equation}

\vspace{0.1cm}

Three cases can occur depending on the value of $Q_t$:

\begin{itemize}
    \item [$\bullet$] $\underline{Q_t > 0}$: The RL agent \textit{buys} shares on the stock market, by posting new \textit{bid} orders on the order book.
    \item [$\bullet$] $\underline{Q_t < 0}$: The RL agent \textit{sells} shares on the stock market, by posting new \textit{ask} orders on the order book.
    \item [$\bullet$] $\underline{Q_t = 0}$: The RL agent \textit{holds}, meaning that it does not buy nor sell any shares on the stock market.
\end{itemize}

Actually, the real actions occurring in the scope of a trading activity are the orders posted on the order book. The RL agent is assumed to communicate with an external module responsible for the synthesis of these true actions according to the value of $Q_t$: the \textit{trading execution system}. Despite being out of the scope of this paper, it should be mentioned that multiple execution strategies can be considered depending on the general trading purpose.\\

The trading actions have an impact on the two components of the portfolio value, namely the cash and share values. Assuming that the trading actions occur close to the market closure at price $p_t \simeq p^C_t$, the updates of these components are governed by the following equations:

\begin{equation}
    \label{EquationCashValue}
    v^c_{t+1} = v^c_t - Q_t\ p_t
\end{equation}

\begin{equation}
    \label{EquationSharesValue}
    v^s_{t+1} = \underbrace{(n_t + Q_t)}_{n_{t+1}}\ p_{t+1}
\end{equation}

\noindent with $n_t \in \mathbb{Z}$ being the number of shares owned by the trading agent at time step $t$. In the scope of this research paper, negative values are allowed for this quantity. Despite being surprising at first glance, a negative number of shares simply corresponds to shares borrowed and sold, with the obligation to repay the lender in shares in the future. Such a mechanism is particularly interesting as it introduces new possibilities for the trading agent.\\

Two important constraints are assumed concerning the quantity of traded shares $Q_t$. Firstly, contrarily to the share value $v^s_t$ which can be both positive or negative, the cash value $v^c_t$ has to remain positive for every trading time steps $t$. This constraint imposes an upper bound on the number of shares that the trading agent is capable of purchasing, this volume of shares being easily derived from Equation \ref{EquationCashValue}. Secondly, there exists a risk associated with the impossibility to repay the share lender if the agent suffers significant losses. To prevent such a situation from happening, the cash value $v^c_t$ is constrained to be sufficiently large when a negative number of shares is owned, in order to be able to get back to a neutral position ($n_t = 0$). A maximum relative change in prices, expressed in \% and denoted $\epsilon \in \mathbb{R}^+$, is assumed by the RL agent prior to the trading activity. This parameter corresponds to the maximum market daily evolution supposed by the agent over the entire trading horizon, so that the trading agent should always be capable of paying back the share lender as long as the market variation remains below this value. Therefore, the constraints acting upon the RL actions at time step $t$ can be mathematically expressed as follows:

\begin{equation}
    \label{EquationConstraint1}
    v^c_{t+1} \geq 0
\end{equation}

\begin{equation}
    \label{EquationConstraint2}
    v^c_{t+1} \ge - n_{t+1}\ p_t\ (1+\epsilon)
\end{equation}

\noindent with the following condition assumed to be satisfied:

\begin{equation}
    \label{EquationNew1}
    \left| \frac{p_{t+1} - p_t}{p_t}\right| \le \epsilon 
\end{equation}

\vspace{0.5cm}

\underline{\textbf{Trading costs consideration:}}\\
\indent Actually, the modelling represented by Equation \ref{EquationCashValue} is inaccurate and will inevitably lead to unrealistic results. Indeed, whenever simulating trading activities, the trading costs should not be neglected. Such omission is generally misleading as a trading strategy, highly profitable in simulations, may be likely to generate large losses in real trading situations due to these trading costs, especially when the trading frequency is high. The trading costs can be subdivided into two categories. On the one hand, there are explicit costs which are induced by transaction costs and taxes. On the other hand, there are implicit costs, called slippage costs, which are composed of three main elements and are associated to some of the dynamics of the trading environment. The different slippage costs are detailed hereafter:

\begin{itemize}
    \item [$\bullet$] \textbf{Spread costs:} These costs are related to the difference between the minimum ask price $p_{min}^{ask}$ and the maximum bid price $p_{max}^{bid}$, called the \textit{spread}. Because the complete state of the order book is generally too complex to efficiently process or even not available, the trading decisions are mostly based on the middle price $p^{mid} = (p_{max}^{bid} + p_{min}^{ask})/2$. However, a buying (selling) trade issued at $p^{mid}$ inevitably occurs at a price $p \geq p_{min}^{ask}$ ($p \leq p_{max}^{bid}$). Such costs are all the more significant that the stock market liquidity is low compared to the volume of shares traded.
    
    \item [$\bullet$] \textbf{Market impact costs:} These costs are induced by the impact of the trader's actions on the market. Each trade (both buying and selling orders) is potentially capable of influencing the price. This phenomenon is all the more important that the stock market liquidity is low with respect to the volume of shares traded.
    
    \item [$\bullet$] \textbf{Timing costs:} These costs are related to the time required for a trade to physically happen once the trading decision is made, knowing that the market price is continuously evolving. The first cause is the inevitable latency which delays the posting of the orders on the market order book. The second cause is the intentional delays generated by the trading execution system. For instance, a large trade could be split into multiple smaller trades spread over time in order to limit the market impact costs.
\end{itemize}

An accurate modelling of the trading costs is required to realistically reproduce the dynamics of the real trading environment. While explicit costs are relatively easy to take into account, the valid modelling of slippage costs is a truly complex task. In this research paper, the integration of both costs into the RL environment is performed through a heuristic. When a trade is executed, a certain amount of capital equivalent to a percentage $C$ of the amount of money invested is lost. This parameter was realistically chosen equal to 0.1\% in the forthcoming simulations.\\

Practically, these trading costs are directly withdrawn from the trading agent cash. Following the heuristic previously introduced, Equations \ref{EquationCashValue} can be re-expressed with a corrective term modelling the trading costs:

\begin{equation}
    \label{EquationCashValueCorrected}
    v^c_{t+1} = v^c_t - Q_t\ p_t - \underbrace{C\ |Q_{t}|\ p_t}_{\text{Trading costs}}
\end{equation}

Moreover, the trading costs have to be properly considered in the constraint expressed in Equation \ref{EquationConstraint2}. Indeed, the cash value $v^c_t$ should be sufficiently large to get back to a neutral position ($n_t = 0$) when the maximum market variation $\epsilon$ occurs, the trading costs being included. Consequently, Equation \ref{EquationConstraint2} is re-expressed as follows:

\begin{equation}
    \label{EquationConstraint2Bis}
    v^c_{t+1} \ge - n_{t+1}\ p_t\ (1+\epsilon)(1+C)
\end{equation}

Eventually, the RL action space $\mathcal{A}$ can be defined as the discrete set of acceptable values for the quantity of traded shares $Q_t$. Derived in detail in \ref{AppendixA}, the RL action space $\mathcal{A}$ is mathematically expressed as the following:

\begin{equation}
    \label{EquationActionSpace}
    \mathcal{A} = \{Q_t \in \mathbb{Z} \cap [\underline{Q_t},\ \overline{Q_t}]\}
\end{equation}

\noindent where:

\begin{itemize}
    \item [$\bullet$] $\overline{Q_t} = \frac{v^c_t}{p_t\ (1+C)}$
    \item [$\bullet$] $\underline{Q_t} = \left\{\begin{matrix}
                                                    \frac{\Delta_t}{p_t\ \epsilon(1 + C)}\ \ \ \ \ \ \ \ \ \text{if } \Delta_t \ge 0\\ 
                                                    \frac{\Delta_t}{p_t\ (2C + \epsilon(1 + C))} \ \ \ \text{if } \Delta_t < 0
                                                \end{matrix}
                                             \right.$\\
                        with $\Delta_t = -v^c_t-n_t\ p_t\ (1+\epsilon)(1+C)$.
\end{itemize}

\vspace{0.5cm}

\underline{\textbf{Action space reduction:}}\\
\indent In the scope of this scientific research paper, the action space $\mathcal{A}$ is reduced in order to lower the complexity of the algorithmic trading problem. The reduced action space is composed of only two RL actions which can be mathematically expressed as the following:

\begin{equation}
    \label{EquationActionsReduced}
    a_t = Q_t \in \{Q^{Long}_t,\ Q^{Short}_t\}
\end{equation}

The first RL action $Q^{Long}_t$ maximises the number of shares owned by the trading agent, by converting as much cash value $v^c_t$ as possible into share value $v^s_t$. It can be mathematically expressed as follows:

\begin{equation}
    \label{EquationActionLong}
    Q^{Long}_t =
    \begin{cases}
        \left\lfloor \frac{v^c_t}{p_t\ (1+C)} \right\rfloor & \text{if } a_{t-1} \neq Q^{Long}_{t-1}, \\
        0 & \text{otherwise.}
    \end{cases}
\end{equation}

The action $Q^{Long}_t$ is always valid as it is obviously included into the original action space $\mathcal{A}$ defined by Equation \ref{EquationActionSpace}. As a result of this action, the trading agent owns a number of shares $N^{Long}_t = n_t + Q^{Long}_t$. On the contrary, the second RL action, designated by $Q^{Short}_t$, converts share value $v^s_t$ into cash value $v^c_t$, such that the RL agent owns a number of shares equal to $-N^{Long}_t$. This operation can be mathematically expressed as the following:

\begin{equation}
    \label{EquationActionShort}
    \widehat{Q^{Short}_t} =
    \begin{cases}
        -2n_t - \left\lfloor \frac{v^c_t}{p_t\ (1+C)} \right\rfloor & \text{if } a_{t-1} \neq Q^{Short}_{t-1}, \\
        0 & \text{otherwise.}
    \end{cases}
\end{equation}

However, the action $\widehat{Q^{Short}_t}$ may violate the lower bound $\underline{Q_t}$ of the action space $\mathcal{A}$ when the price significantly increases over time. Eventually, the second RL action $Q^{Short}_t$ is expressed as follows:

\begin{equation}
    Q^{Short}_t = \text{max} \left\{\widehat{Q^{Short}_t},\ \underline{Q_t} \right\}
\end{equation}

To conclude this subsection, it should be mentioned that the two reduced RL actions are actually related to the next trading position of the agent, designated as $P_{t+1}$. Indeed, the first action $Q^{Long}_t$ induces a \textit{long} trading position because the number of owned shares is positive. On the contrary, the second action $Q^{Short}_t$ always results in a number of shares which is negative, which is generally referred to as a \textit{short} trading position in finance.

\subsubsection{RL rewards}
\label{SectionRLRewards}

For this algorithmic trading problem, a natural choice for the RL rewards is the strategy daily returns. Intuitively, it makes sense to favour positive returns which are an evidence of a profitable strategy. Moreover, such quantity has the advantage of being independent of the number of shares $n_t$ currently owned by the agent. This choice is also motivated by the fact that it allows to avoid a sparse reward setup, which is more complex to deal with. The RL rewards can be mathematically expressed as the following:

\begin{equation}
\label{EquationReward}
    r_t = \frac{v_{t+1} - v_t}{v_t}
\end{equation}

\subsection{Objective}
\label{SectionObjective}

Objectively assessing the performance of a trading strategy is a tricky task, due to the numerous quantitative and qualitative factors to consider. Indeed, a well-performing trading strategy is not simply expected to generate profit, but also to efficiently mitigate the risk associated with the trading activity. The balance between these two goals varies depending on the trading agent profile and its willingness to take extra risks. Although intuitively convenient, maximising the profit generated by a trading strategy is a necessary but not sufficient objective. Instead, the core objective of a trading strategy is the maximisation of the \textit{Sharpe ratio}, a performance indicator widely used in the fields of finance and algorithmic trading. It is particularly well suited for the performance assessment task as it considers both the generated profit and the risk associated with the trading activity. Mathematically, the Sharpe ratio $S_r$ is expressed as the following:

\begin{equation}
    \label{EquationSharpeRatio}
    S_r = \frac{\mathbb{E}[R_s - R_f]}{\sigma_r} = \frac{\mathbb{E}[R_s - R_f]}{\sqrt{\text{var}[R_s - R_f]}} \simeq \frac{\mathbb{E}[R_s]}{\sqrt{\text{var}[R_s]}}
\end{equation}

\noindent where:

\begin{itemize}
    \item [$\bullet$] $R_s$ is the trading strategy return over a  certain time period, modelling its profitability.
    \item [$\bullet$] $R_f$ is the risk-free return, the expected return from a totally safe investment (negligible).
    \item [$\bullet$] $\sigma_r$ is the standard deviation of the trading strategy excess return $R_s - R_f$, modelling its riskiness.
\end{itemize}

In order to compute the Sharpe ratio $S_r$ in practice, the daily returns achieved by the trading strategy are firstly computed using the formula $\rho_t = (v_t - v_{t-1})/v_{t-1}$. Then, the ratio between the returns mean and standard deviation is evaluated. Finally, the annualised Sharpe ratio is obtained by multiplying this value by the square root of the number of trading days in a year (252).\\

Moreover, a well-performing trading strategy should ideally be capable of achieving acceptable performance on diverse markets presenting very different patterns. For instance, the trading strategy should properly handle both bull and bear markets (respectively strong increasing and decreasing price trends), with different levels of volatility. Therefore, the research paper's core objective is the development of a novel trading strategy based on DRL techniques to maximise the average Sharpe ratio computed on the entire set of existing stock markets.\\

Despite the fact that the ultimate objective is the maximisation of the Sharpe ratio, the DRL algorithm adopted in this scientific paper actually maximises the expected discounted sum of rewards (daily returns) over an infinite time horizon. This optimisation criterion, which does not exactly corresponds to maximising profits but is very close to that, can in fact be seen as a relaxation of the Sharpe ratio criterion. A future interesting research direction would be to narrow the gap between these two objectives.

\section{Deep reinforcement learning algorithm design}
\label{SectionDRLAlgorithm}

In this section, a novel DRL algorithm is designed to solve the algorithmic trading problem previously introduced. The resulting trading strategy, denominated the Trading Deep Q-Network algorithm (TDQN), is inspired from the successful DQN algorithm presented in \cite{Mnih2013} and is significantly adapted to the specific decision-making problem at hand. Concerning the training of the RL agent, artificial trajectories are generated from a limited set of stock market historical data.

\subsection{Deep Q-Network algorithm}
\label{SectionDQN}

The Deep Q-Network algorithm, generally referred to as DQN, is a DRL algorithm capable of successfully learning control policies from high-dimensional sensory inputs. It is in a way the successor of the popular Q-learning algorithm introduced in \cite{Watkins1992}. This DRL algorithm is said to be \textit{model-free}, meaning that a complete model of the environment is not required and that trajectories are sufficient. Belonging to the \textit{Q-learning} family of algorithms, it is based on the learning of an approximation of the state-action value function, which is represented by a DNN. In such context, learning the Q-function amounts to learning the parameters $\theta$ of this DNN. Finally, the DQN algorithm is said to be \textit{off-policy} as it exploits in batch mode previous experiences $e_t=(s_t, a_t, r_t, s_{t+1})$ collected at any point during training.\\

For the sake of brevity, the DQN algorithm is illustrated in Figure \ref{DQNAlgorithm}, but is not extensively presented in this paper. Besides the original publications (\cite{Mnih2013} and \cite{Mnih2015}), there exists a great scientific literature around this algorithm, see for instance \cite{Hasselt2016}, \cite{Wang2016}, \cite{Schaul2016}, \cite{Bellemare2017}, \cite{Fortunato2018} and \cite{Hessel2018}. Concerning DL techniques, interesting resources are \cite{LeCun2015}, \cite{Goodfellow2015} and \cite{Goodfellow2016}. For more information about RL, the reader can refer to the following textbooks and surveys: \cite{Sutton2018}, \cite{Szepesvri2010}, \cite{Busoniu2010}, \cite{Arulkumaran2017} and \cite{Shao2019}.

\begin{figure}
    \centering
    \includegraphics[scale=0.29]{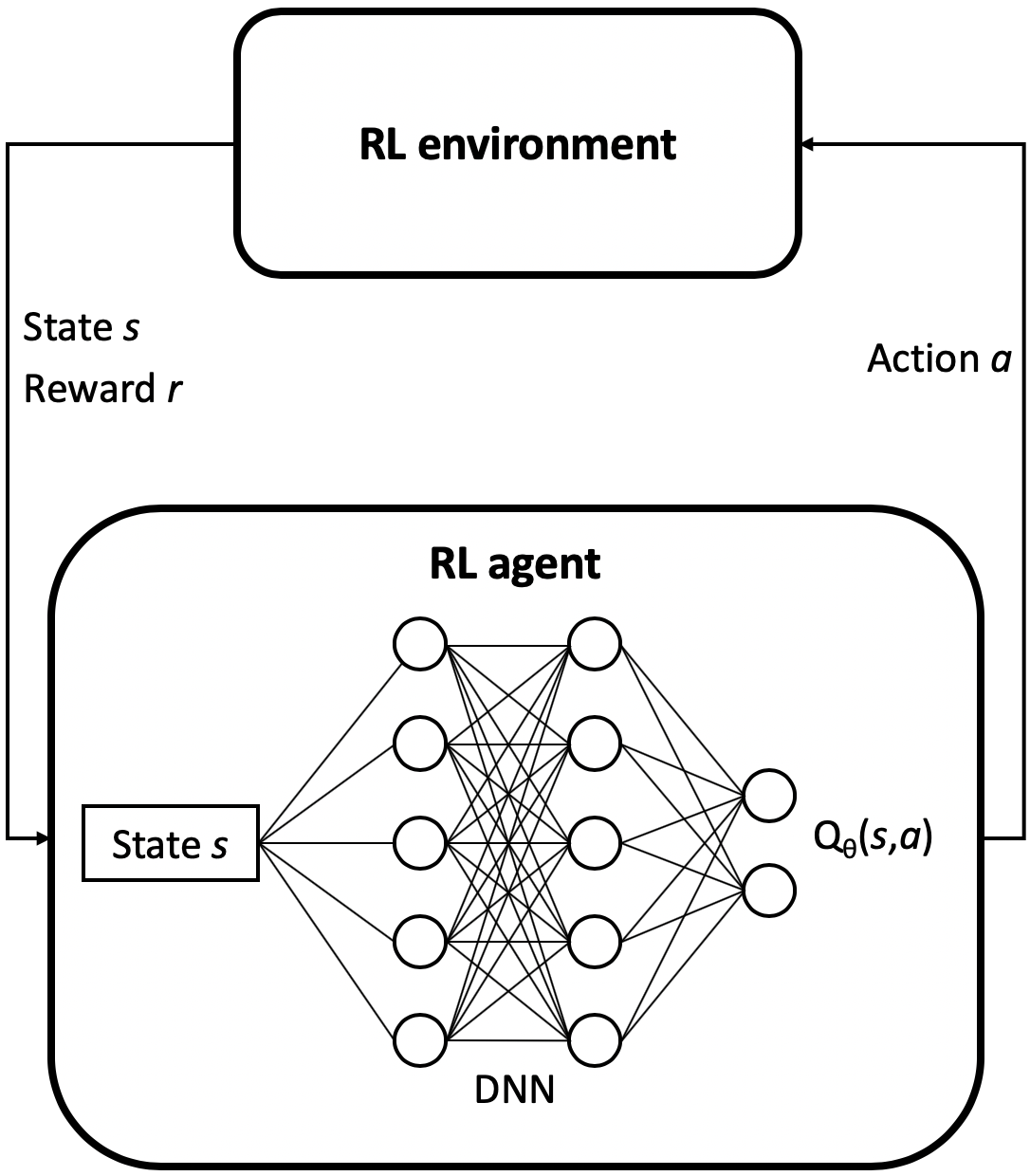}
    \caption{Illustration of the DQN algorithm}
    \label{DQNAlgorithm}
\end{figure}

\subsection{Artificial trajectories generation}
\label{SectionTrajectories}

In the scope of the algorithmic trading problem, a complete model of the environment $\mathcal{E}$ is not available. The training of the TDQN algorithm is entirely based on the generation of artificial trajectories from a limited set of stock market historical daily OHLCV data. A trajectory $\tau$ is defined as a sequence of observations $o_t \in \mathcal{O}$, actions $a_t \in \mathcal{A}$ and rewards $r_t$ from an RL agent for a certain number $T$ of trading time steps $t$:

\begin{equation*}
    \tau = \Big(\{o_0, a_0, r_0\}, \{o_1, a_1, r_1\}, ..., \{o_T, a_T, r_T\}\Big)
\end{equation*}

Initially, although the environment $\mathcal{E}$ is unknown, one disposes of a single real trajectory, corresponding to the historical behaviour of the stock market, i.e. the particular case of the RL agent being inactive. This original trajectory is composed of the historical prices and volumes together with long actions executed by the RL agent with no money at its disposal, to represent the fact that no shares are actually traded. For this algorithmic trading problem, new fictive trajectories are then artificially generated from this unique true trajectory to simulate interactions with the environment $\mathcal{E}$. The historical stock market behaviour is simply considered unaffected by the new actions performed by the trading agent. The artificial trajectories generated are simply composed of the sequence of historical real observations associated with various sequences of trading actions from the RL agent. For such practice to be scientifically acceptable and lead to realistic simulations, the trading agent should not be able to influence the stock market behaviour. This assumption generally holds when the number of shares traded by the trading agent is low with respect to the liquidity of the stock market.\\

In addition to the generation of artificial trajectories just described, a trick is employed to slightly improve the exploration of the RL agent. It relies on the fact that the reduced action space $\mathcal{A}$ is composed of only two actions: long ($Q^{Long}_t$) and short ($Q^{Short}_t$). At each trading time step $t$, the chosen action $a_t$ is executed on the trading environment $\mathcal{E}$ and the opposite action $a^-_t$ is executed on a copy of this environment $\mathcal{E}^-$. Although this trick does not completely solve the challenging exploration/exploitation trade-off, it enables the RL agent to continuously explore at a small extra computational cost.

\subsection{Diverse modifications and improvements}
\label{SectionAlgorithmImprovements}

The DQN algorithm was chosen as starting point for the novel DRL trading strategy developed, but was significantly adapted to the specific algorithmic trading decision-making problem at hand. The diverse modifications and improvements, which are mainly based on the numerous simulations performed, are summarised hereafter:

\begin{itemize}
    \item [$\bullet$] \textbf{Deep neural network architecture:} The first difference with respect to the classical DQN algorithm is the architecture of the DNN approximating the action-value function $Q(s,a)$. Due to the different nature of the input (time-series instead of raw images), the convolutional neural network (CNN) has been replaced by a classical feedforward DNN with some leaky rectified linear unit (Leaky ReLU) activation functions.
    
    \item [$\bullet$] \textbf{Double DQN:} The DQN algorithm suffers from substantial overestimations, this overoptimism harming the algorithm performance. In order to reduce the impact of this undesired phenomenon, the article \cite{Hasselt2016} presents the double DQN algorithm which is based on the decomposition of the target max operation into both action selection and action evaluation.
    
    \item [$\bullet$] \textbf{ADAM optimiser:} The classical DQN algorithm implements the RMSProp optimiser. However, the ADAM optimiser, introduced in \cite{Kingma2015}, experimentally proves to improve both the training stability and the convergence speed of the DRL algorithm.
    
    \item [$\bullet$] \textbf{Huber loss:} While the classical DQN algorithm implements a mean squared error (MSE) loss, the Huber loss experimentally improves the stability of the training phase. Such observation is explained by the fact that the MSE loss significantly penalises large errors, which is generally desired but has a negative side-effect for the DQN algorithm because the DNN is supposed to predict values that depend on its own input. This DNN should not radically change in a single training update because this would also lead to a significant change in the target, which could actually result in a larger error. Ideally, the update of the DNN should be performed in a slower and more stable manner. On the other hand, the mean absolute error (MAE) has the drawback of not being differentiable at 0. A good trade-off between these two losses is the Huber loss $H$:
    
    \begin{equation}
        \label{EquationHuberLoss}
        H(x) = \left\{
                    \begin{array}{ll}
                        \frac{1}{2}x^2 & \mbox{if } |x| \leq 1, \\
                        |x| - \frac{1}{2} & \mbox{otherwise.}
                    \end{array}
                \right.
    \end{equation}
    
    \begin{figure}[H]
        \centering
        \includegraphics[scale=0.44, trim={0.5cm 0cm 0.5cm 1.3cm}, clip]{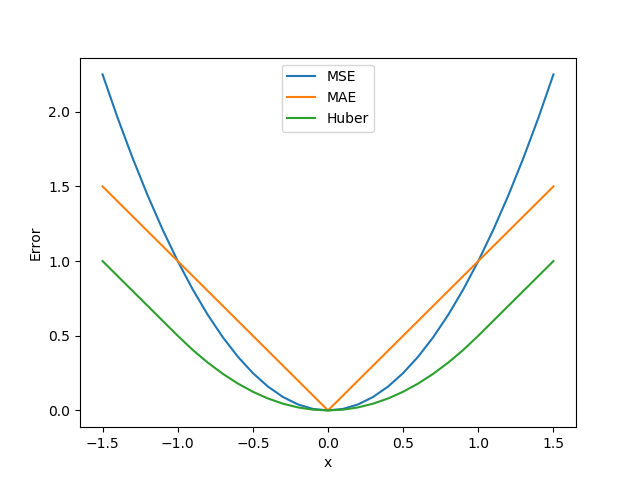}
        \caption{Comparison of the MSE, MAE and Huber losses}
        \label{Huber}
    \end{figure}
    
    \item [$\bullet$] \textbf{Gradient clipping:} The gradient clipping technique is implemented in the TDQN algorithm to solve the gradient exploding problem which induces significant instabilities during the training of the DNN.
    
    \item [$\bullet$] \textbf{Xavier initialisation:} While the classical DQN algorithm simply initialises the DNN weights randomly, the Xavier initialisation is implemented to improve the algorithm convergence. The idea is to set the initial weights so that the gradients variance remains constant across the DNN layers.
    
    \item [$\bullet$] \textbf{Batch normalisation layers:} This DL technique, introduced by \cite{Ioffe2015}, consists in normalising the input layer by adjusting and scaling the activation functions. It brings many benefits including a faster and more robust training phase as well as an improved generalisation.
    
    \item [$\bullet$] \textbf{Regularisation techniques:} Because a strong tendency to overfit was observed during the first experiments with the DRL trading strategy, three regularisation techniques are implemented: \textit{Dropout}, \textit{L2 regularisation} and \textit{Early Stopping}.
    
    \item[$\bullet$] \textbf{Preprocessing and normalisation:} The training loop of the TDQN algorithm is preceded by both a preprocessing and a normalisation operation of the RL observations $o_t$. Firstly, because the high-frequency noise present in the trading data was experimentally observed to lower the algorithm generalisation, a low-pass filtering operation is executed. However, such a preprocessing operation has a cost as it modifies or even destroys some potentially useful trading patterns and introduces a non-negligible lag. Secondly, the resulting data are transformed in order to convey more meaningful information about market movements. Typically, the daily evolution of prices is considered rather than the raw prices. Thirdly, the remaining data are normalised.
    
    \item [$\bullet$] \textbf{Data augmentation techniques:} A key challenge of this algorithmic trading problem is the limited amount of available data, which are in addition generally of poor quality. As a counter to this major problem, several data augmentation techniques are implemented: signal shifting, signal filtering and artificial noise addition. The application of such data augmentation techniques will artificially generate new trading data which are slightly different but which result in the same financial phenomena.
\end{itemize}

Finally, the algorithm underneath the TDQN trading strategy is depicted in detail in Algorithm \ref{TDQN}.

\begin{algorithm*}
\small
\caption{TDQN algorithm}
\begin{algorithmic} 
\STATE Initialise the experience replay memory $M$ of capacity $C$.
\STATE Initialise the main DNN weights $\theta$ (Xavier initialisation).
\STATE Initialise the target DNN weights $\theta^- = \theta$.
\FOR{episode = 1 \TO N}
    \STATE Acquire the initial observation $o_1$ from the environment $\mathcal{E}$ and preprocess it.
    \FOR{t = 1 \TO T}
        \STATE With probability $\epsilon$, select a random action $a_t$ from $\mathcal{A}$.
        \STATE Otherwise, select $a_t = \mbox{arg max}_{a \in \mathcal{A}} Q(o_t, a; \theta)$.
        \STATE Copy the environment $\mathcal{E}^- = \mathcal{E}$.
        \STATE Interact with the environment $\mathcal{E}$ (action $a_t$) and get the new observation $o_{t+1}$ and reward $r_t$.
        \STATE Perform the same operation on $\mathcal{E}^-$ with the opposite action $a^-_t$, getting $o^-_{t+1}$ and $r^-_t$.
        \STATE Preprocess both new observations $o_{t+1}$ and $o^-_{t+1}$.
        \STATE Store both experiences $e_t = (o_t, a_t, r_t, o_{t+1})$ and $e^-_t = (o_t, a^-_t, r^-_t, o^-_{t+1})$ in $M$.
        \IF{t \% T' = 0}
            \STATE Randomly sample from $M$ a minibatch of $N_e$ experiences $e_i = (o_i, a_i, r_i, o_{i+1})$.
            \STATE Set $y_i = \left\{ \begin{array}{ll}
                             r_i & \mbox{if the state } s_{i+1} \mbox{ is terminal,} \\
                             r_i + \gamma \ Q(o_{i+1}, \mbox{arg max}_{a \in \mathcal{A}} Q(o_{i+1}, a; \theta); \theta^-) & \mbox{otherwise.}
                          \end{array}
                         \right.$
            \STATE Compute and clip the gradients based on the Huber loss $H(y_i,\ Q(o_i, a_i; \theta))$.
            \STATE Optimise the main DNN parameters $\theta$ based on these clipped gradients.
            \STATE Update the target DNN parameters $\theta^- = \theta$ every $N^-$ steps.\\ 
        \ENDIF
        \STATE Anneal the $\epsilon$-Greedy exploration parameter $\epsilon$.
    \ENDFOR
\ENDFOR
\end{algorithmic} 
\label{TDQN}
\end{algorithm*}

\section{Performance assessment}
\label{SectionPerformanceAssessment}

An accurate performance evaluation approach is capital in order to produce meaningful results. As previously hinted, this procedure is all the more critical because there has been a real lack of a proper performance assessment methodology in the algorithmic trading field. In this section, a novel, more reliable methodology is presented to objectively assess the performance of algorithmic trading strategies, including the TDQN algorithm.

\subsection{Testbench}
\label{SectionTestbench}

In the literature, the performance of a trading strategy is generally assessed on a single instrument (stock market or others) for a certain period of time. Nevertheless, the analysis resulting from such a basic approach should not be entirely trusted, as the trading data could have been specifically selected so that a trading strategy looks profitable, even though it is not the case in general. To eliminate such bias, the performance should ideally be assessed on multiple instruments presenting diverse patterns. Aiming to produce trustful conclusions, this research paper proposes a testbench composed of 30 stocks presenting diverse characteristics (sectors, regions, volatility, liquidity, etc.). The testbench is depicted in Table \ref{testbench}. To avoid any confusion, the official reference for each stock (ticker) is specified in parentheses. To avoid any ambiguities concerning the training and evaluation protocols, it should be mentioned that a new trading strategy is trained for each stock included in the testbench. Nevertheless, for the sake of generality, all the algorithm hyperparameters remain unchanged over the entire testbench.\\

\begin{table*}
  \small
  \caption{Performance assessment testbench}
  \label{testbench}
  \centering
  \begin{tabular}{cccc}
    \toprule
    \multirow{2}{*}{\textbf{Sector}} & \multicolumn{3}{c}{\textbf{Region}}\\
    \cmidrule(r){2-4}
    & \textit{American} & \textit{European} & \textit{Asian} \\
    \midrule
    \multirow{3}{*}{\textit{Trading index}} & Dow Jones (DIA) & FTSE 100 (EZU) & Nikkei 225 (EWJ) \\
    & S\&P 500 (SPY) & & \\
    & NASDAQ (QQQ) & & \\
    \cmidrule(r){1-4}
    \multirow{5}{*}{\textit{Technology}} & Apple (AAPL) & Nokia (NOK) & Sony (6758.T) \\
    & Google (GOOGL) & Philips (PHIA.AS) & Baidu (BIDU) \\
    & Amazon (AMZN) & Siemens (SIE.DE) & Tencent (0700.HK) \\
    & Facebook (FB) & & Alibaba (BABA) \\
    & Microsoft (MSFT) & & \\
    & Twitter (TWTR) & & \\
    \cmidrule(r){1-4}
    \textit{Financial services} & JPMorgan Chase (JPM) & HSBC (HSBC) & CCB (0939.HK) \\
    \cmidrule(r){1-4}
    \textit{Energy} & ExxonMobil (XOM) & Shell (RDSA.AS) & PetroChina (PTR) \\
    \cmidrule(r){1-4}
    \textit{Automotive} & Tesla (TSLA) & Volkswagen (VOW3.DE) & Toyota (7203.T) \\
    \cmidrule(r){1-4}
    \textit{Food} & Coca Cola (KO) & AB InBev (ABI.BR) & Kirin (2503.T) \\
    \bottomrule
  \end{tabular}
\end{table*}

Regarding the trading horizon, the eight years preceding the publication year of the research paper are selected to be representative of the current market conditions. Such a short-time period could be criticised because it may be too limited to be representative of the entire set of financial phenomena. For instance, the financial crisis of 2008 is rejected, even though it could be interesting to assess the robustness of trading strategies with respect to such an extraordinary event. However, this choice was motivated by the fact that a shorter trading horizon is less likely to contain significant market regime shifts which would seriously harm the training stability of the trading strategies. Finally, the trading horizon of eight years is divided into both training and test sets as follows:

\begin{itemize}
    \item [$\bullet$] \textbf{Training set:} 01/01/2012 $\rightarrow$ 31/12/2017.
    \item [$\bullet$] \textbf{Test set:} 01/01/2018 $\rightarrow$ 31/12/2019.
\end{itemize}

A validation set is also considered as a subset of the training set for the tuning of the numerous TDQN algorithm hyperparameters. Note that the RL policy DNN parameters $\theta$ are fixed during the execution of the trading strategy on the entire test set, meaning that the new experiences acquired are not valued for extra training. Nevertheless, such practice constitutes an interesting future research direction.\\

To end this subsection, it should be noted that the proposed testbench could be improved thanks to even more diversification. The obvious addition would be to include more stocks with different financial situations and properties. Another interesting addition would be to consider different training/testing time periods while excluding the significant market regime shifts. Nevertheless, this last idea was discarded in this scientific article due to the important time already required to produce results for the proposed testbench.

\subsection{Benchmark trading strategies}
\label{SectionBenchmark}

In order to properly assess the strengths and weaknesses of the TDQN algorithm, some benchmark algorithmic trading strategies were selected for comparison purposes. Only the classical trading strategies commonly used in practice were considered, excluding for instance strategies based on DL techniques or other advanced approaches. Despite the fact that the TDQN algorithm is an active trading strategy, both passive and active strategies are taken into consideration. For the sake of fairness, the strategies share the same input and output spaces presented in Section \ref{SectionRLActions} ($\mathcal{O}$ and $\mathcal{A}$). The following list summarises the benchmark strategies selected:

\begin{itemize}
    \item [$\bullet$] Buy and hold (B\&H).
    \item [$\bullet$] Sell and hold (S\&H).
    \item [$\bullet$] Trend following with moving averages (TF).
    \item [$\bullet$] Mean reversion with moving averages (MR).
\end{itemize}

For the sake of brevity, a detailed description of each strategy is not provided in this research paper. The reader can refer to \cite{Chan2009}, \cite{Chan2013} or \cite{Narang2009} for more information. The first two benchmark trading strategies (B\&H and S\&H) are said to be passive, as there are no changes in trading position over the trading horizon. On the contrary, the other two benchmark strategies (TF and MR) are active trading strategies, issuing multiple changes in trading positions over the trading horizon. On the one hand, a trend following strategy is concerned with the identification and the follow-up of significant market trends, as depicted in Figure \ref{TrendFollowing}. On the other hand, a mean reversion strategy, illustrated in Figure \ref{MeanReversion}, is based on the tendency of a stock market to get back to its previous average price in the absence of clear trends. By design, a trend following strategy generally makes a profit when a mean reversion strategy does not, the opposite being true as well. This is due to the fact that these two families of trading strategies adopt opposite positions: a mean reversion strategy always denies and goes against the trends while a trend following strategy follows the movements.

\vspace{0.3cm}

\begin{figure}[H]
    \centering
    \includegraphics[scale=0.28]{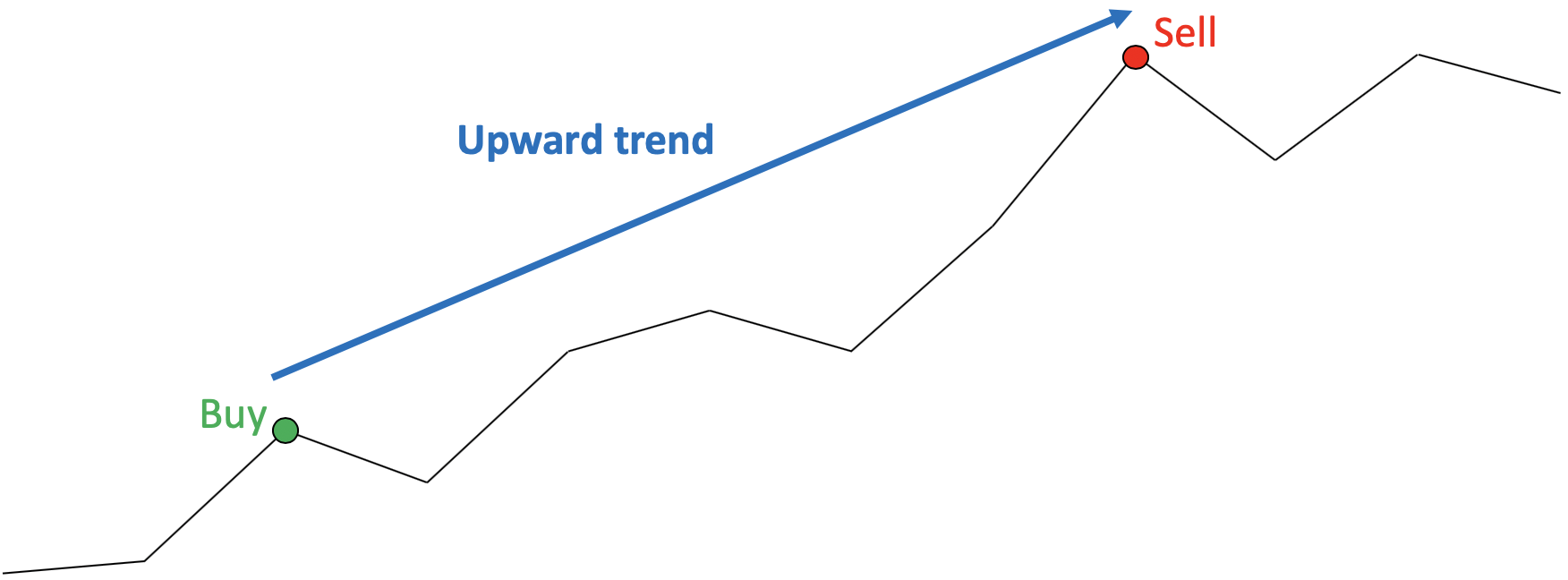}
    \caption{Illustration of a typical trend following trading strategy}
    \label{TrendFollowing}
\end{figure}

\vspace{0.3cm}

\begin{figure}[H]
    \centering
    \includegraphics[scale=0.34]{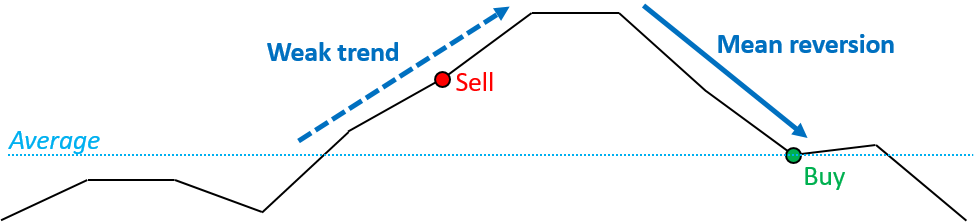}
    \caption{Illustration of a typical mean reversion trading strategy}
    \label{MeanReversion}
\end{figure}

\subsection{Quantitative performance assessment}
\label{SectionQuantitativePerformance}

\begin{table*}
  \small
  \caption{Quantitative performance assessment indicators}
  \label{performanceIndicators}
  \centering
  \begin{tabular}{ll}
    \toprule
    \textbf{Performance indicator} & \textbf{Description} \\
    \midrule
    Sharpe ratio & Return of the trading activity compared to its riskiness.\\
    Profit \& loss & Money gained or lost at the end of the trading activity.\\
    Annualised return & Annualised return generated during the trading activity.\\
    Annualised volatility & Modelling of the risk associated with the trading activity.\\
    Profitability ratio & Percentage of winning trades made during the trading activity.\\
    Profit and loss ratio & Ratio between the trading activity trades average profit and loss.\\
    Sortino ratio & Similar to the Sharpe ratio with the negative risk penalised only.\\
    Maximum drawdown & Largest loss from a peak to a trough during the trading activity.\\
    Maximum drawdown duration & Time duration of the trading activity maximum drawdown.\\
    \bottomrule
  \end{tabular}
\end{table*}

The quantitative performance assessment consists in defining one performance indicator or more to numerically quantify the performance of an algorithmic trading strategy. Because the core objective of a trading strategy is to be profitable, its performance should be linked to the amount of money earned. However, such reasoning omits to consider the risk associated with the trading activity which should be efficiently mitigated. Generally, a trading strategy achieving a small but stable profit is preferred to a trading strategy achieving a huge profit in a very unstable way after suffering from multiple losses. It eventually depends on the investor profile and the willingness to take extra risks to potentially earn more.\\

Multiple performance indicators were selected to accurately assess the performance of a trading strategy. As previously introduced in Section \ref{SectionObjective}, the most important one is certainly the Sharpe ratio. This performance indicator, widely used in the field of algorithmic trading, is particularly informative as it combines both profitability and risk. Besides the Sharpe ratio, this research paper considers multiple other performance indicators to provide extra insights. Table \ref{performanceIndicators} presents the entire set of performance indicators employed to quantify the performance of a trading strategy.\\

Complementarily to the computation of these numerous performance indicators, it is interesting to graphically represent the trading strategy behaviour. Plotting both the stock market price $p_t$ and portfolio value $v_t$ evolutions together with the trading actions $a_t$ issued by the trading strategy seems appropriate to accurately analyse the trading policy. Moreover, such visualisation could also provide extra insights about the performance, the strengths and weaknesses of the strategy analysed.

\section{Results and discussion}
\label{SectionResults}

In this section, the TDQN trading strategy is evaluated following the performance assessment methodology previously described. Firstly, a detailed analysis is performed for both a case that give good results and a case for which the results were mitigated. This highlights the strengths, weaknesses and limitations of the TDQN algorithm. Secondly, the performance achieved by the DRL trading strategy on the entire testbench is summarised and analysed. Finally, some additional discussions about the discount factor parameter, the trading costs influence and the main challenges faced by the TDQN algorithm are provided. The experimental code supporting the results presented is publicly available at the following link:\\

\url{https://github.com/ThibautTheate/An-Application-of-Deep-Reinforcement-Learning-to-Algorithmic-Trading}.

\subsection{Good results - Apple stock}
\label{SectionGoodResults}

The first detailed analysis concerns the execution of the TDQN trading strategy on the Apple stock, resulting in promising results. Similar to many DRL algorithms, the TDQN algorithm is subject to a non-negligible variance. Multiple training experiments with the exact same initial conditions will inevitably lead to slightly different trading strategies of varying performance. As a consequence, both a typical run of the TDQN algorithm and its expected performance are presented hereafter.\\

\textbf{Typical run:} Firstly, Table \ref{PerformanceApple} presents the performance achieved by each trading strategy considered, the initial amount of money being equal to \$100,000. The TDQN algorithm achieves good results from both an earnings and a risk mitigation point of view, clearly outperforming all the benchmark active and passive trading strategies. Secondly, Figure \ref{TestApple} plots both the stock market price $p_t$ and RL agent portfolio value $v_t$ evolutions, together with the actions $a_t$ outputted by the TDQN algorithm. It can be observed that the DRL trading strategy is capable of accurately detecting and benefiting from major trends, while being more hesitant during market behavioural shifts when the volatility increases. It can also be seen that the trading agent generally lags slightly behind the market trends, meaning that the TDQN algorithm learned to be more reactive than proactive for this particular stock. This behaviour is expected with such a limited observation space $\mathcal{O}$ not including the reasons for the future market directions (new product announcement, financial report, macroeconomics, etc.). However, this does not mean that the policies learned are purely reactive. Indeed, it was observed that the RL agent may decide to adapt its trading position before a trend inversion by noticing an increase in volatility, therefore anticipating and being proactive.\\

\begin{table*}
  \small
  \caption{Performance assessment for the Apple stock}
  \label{PerformanceApple}
  \centering
  \begin{tabular}{lccccc}
    \toprule
    \textbf{Performance indicator} & \textbf{B\&H} & \textbf{S\&H} & \textbf{TF} & \textbf{MR} & \textbf{TDQN} \\
    \midrule
    Sharpe ratio & 1.239 & -1.593 & 1.178 & -0.609 & 1.484 \\
    Profit \& loss [\$] & 79823 & -80023 & 68738 & -34630 & 100288 \\
    Annualised return [\%] & 28.86 & -100.00 & 25.97 & -19.09 & 32.81 \\
    Annualised volatility [\%] & 26.62 & 44.39 & 24.86 & 28.33 & 25.69 \\
    Profitability ratio [\%] & 100 & 0.00 & 42.31 & 56.67 & 52.17 \\
    Profit and loss ratio & $\infty$ & 0.00 & 3.182 & 0.492 & 2.958 \\
    Sortino ratio & 1.558 & -2.203 & 1.802 & -0.812 & 1.841 \\
    Max drawdown [\%] & 38.51 & 82.48 & 14.89 & 51.12 & 17.31 \\
    Max drawdown duration [days] & 62 & 250 & 20 & 204 & 25 \\
    \bottomrule
  \end{tabular}
\end{table*}

\begin{figure}
    \centering
    \includegraphics[scale=0.29, trim={0cm 0cm 0cm 0cm}, clip]{Figures/AAPL_Rendering.png}
    \caption{TDQN algorithm execution for the Apple stock (test set)}
    \label{TestApple}
\end{figure}

\textbf{Expected performance:} In order to estimate the expected performance as well as the variance of the TDQN algorithm, the same RL trading agent is trained multiple times. Figure \ref{ExpectedPerformanceApple} plots the averaged (over 50 iterations) performance of the TDQN algorithm for both the training and test sets with respect to the number of training episodes. This expected performance is comparable to the performance achieved during the typical run of the algorithm. It can also be noticed that the overfitting tendency of the RL agent seems to be properly handled for this specific market. Please note that the test set performance being temporarily superior to the training set performance is not a mistake. It simply indicates an easier to trade and more profitable market for the test set trading period for the Apple stock. This example perfectly illustrates a major difficulty of the algorithmic trading problem: the training and test sets do not share the same distributions. Indeed, the distribution of the daily returns is continuously changing, which complicates both the training of the DRL trading strategy and its performance evaluation.

\begin{figure}
    \centering
    \includegraphics[scale=0.6, trim={0.3cm 0.2cm 1cm 1.15cm}, clip]{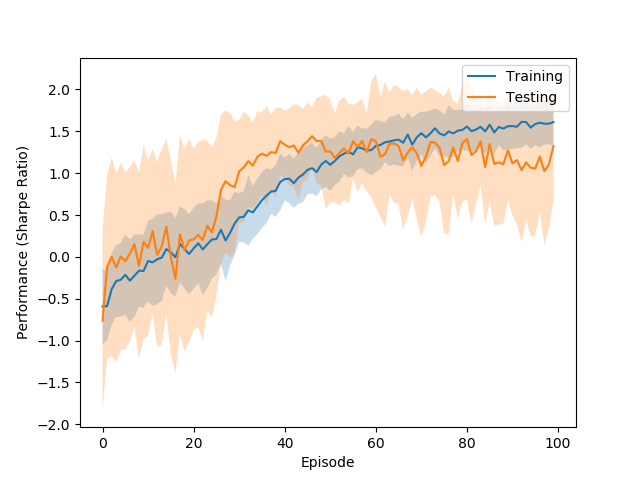}
    \caption{TDQN algorithm expected performance for the Apple stock}
    \label{ExpectedPerformanceApple}
\end{figure}

\subsection{Mitigated results - Tesla stock}
\label{SectionMitigatedResults}

The same detailed analysis is performed on the Tesla stock, which presents very different characteristics compared to the Apple stock, such as a pronounced volatility. In contrast to the promising performance achieved on the previous stock, this case was specifically selected to highlight the limitations of the TDQN algorithm.\\

\begin{table*}
  \small
  \caption{Performance assessment for the Tesla stock}
  \label{PerformanceTesla}
  \centering
  \begin{tabular}{lccccc}
    \toprule
    \textbf{Performance indicator} & \textbf{B\&H} & \textbf{S\&H} & \textbf{TF} & \textbf{MR} & \textbf{TDQN} \\
    \midrule
    Sharpe ratio & 0.508 & -0.154 & -0.987 & 0.358 & 0.261 \\
    Profit \& loss [\$] & 29847 & -29847 & -73301 & 8600 & 98 \\
    Annualised return [\%] & 24.11 & -7.38 & -100.00 & 19.02 & 12.80 \\
    Annualised volatility [\%] & 53.14 & 46.11 & 52.70 & 58.05 & 52.09 \\
    Profitability ratio [\%] & 100 & 0.00 & 34.38 & 67.65 & 38.18 \\
    Profit and loss ratio & $\infty$ & 0.00 & 0.534 & 0.496 & 1.621 \\
    Sortino ratio & 0.741 & -0.205 & -1.229 & 0.539 & 0.359 \\
    Max drawdown [\%] & 52.83 & 54.09 & 79.91 & 65.31 & 58.95 \\
    Max drawdown duration [days] & 205 & 144 & 229 & 159 & 331 \\
    \bottomrule
  \end{tabular}
\end{table*}

\textbf{Typical run:} Similar to the previous analysis, Table \ref{PerformanceTesla} presents the performance achieved by every trading strategies considered, the initial amount of money being equal to \$100,000. The mitigated results achieved by the benchmark active strategies suggest that the Tesla stock is quite difficult to trade, which is partly due to its significant volatility. Even though the TDQN algorithm achieves a positive Sharpe ratio, almost no profit is generated. Moreover, the risk level associated with this trading activity cannot really be considered acceptable. For instance, the maximum drawdown duration is particularly large, which would result in a stressful situation for the operator responsible for the trading strategy. Figure \ref{TestTesla}, which plots both the stock market price $p_t$ and RL agent portfolio value $v_t$ evolutions together with the actions $a_t$ outputted by the TDQN algorithm, confirms this observation. Moreover, it can be clearly observed that the pronounced volatility of the Tesla stock induces a higher trading frequency (changes in trading positions, which correspond to the situation where $a_{t} \neq a_{t-1}$) despite the non-negligible trading costs, which increases even more the riskiness of the DRL trading strategy.\\

\begin{figure}
    \centering
    \includegraphics[scale=0.17, trim={0.1cm 0.1cm 0.1cm 0.1cm}, clip]{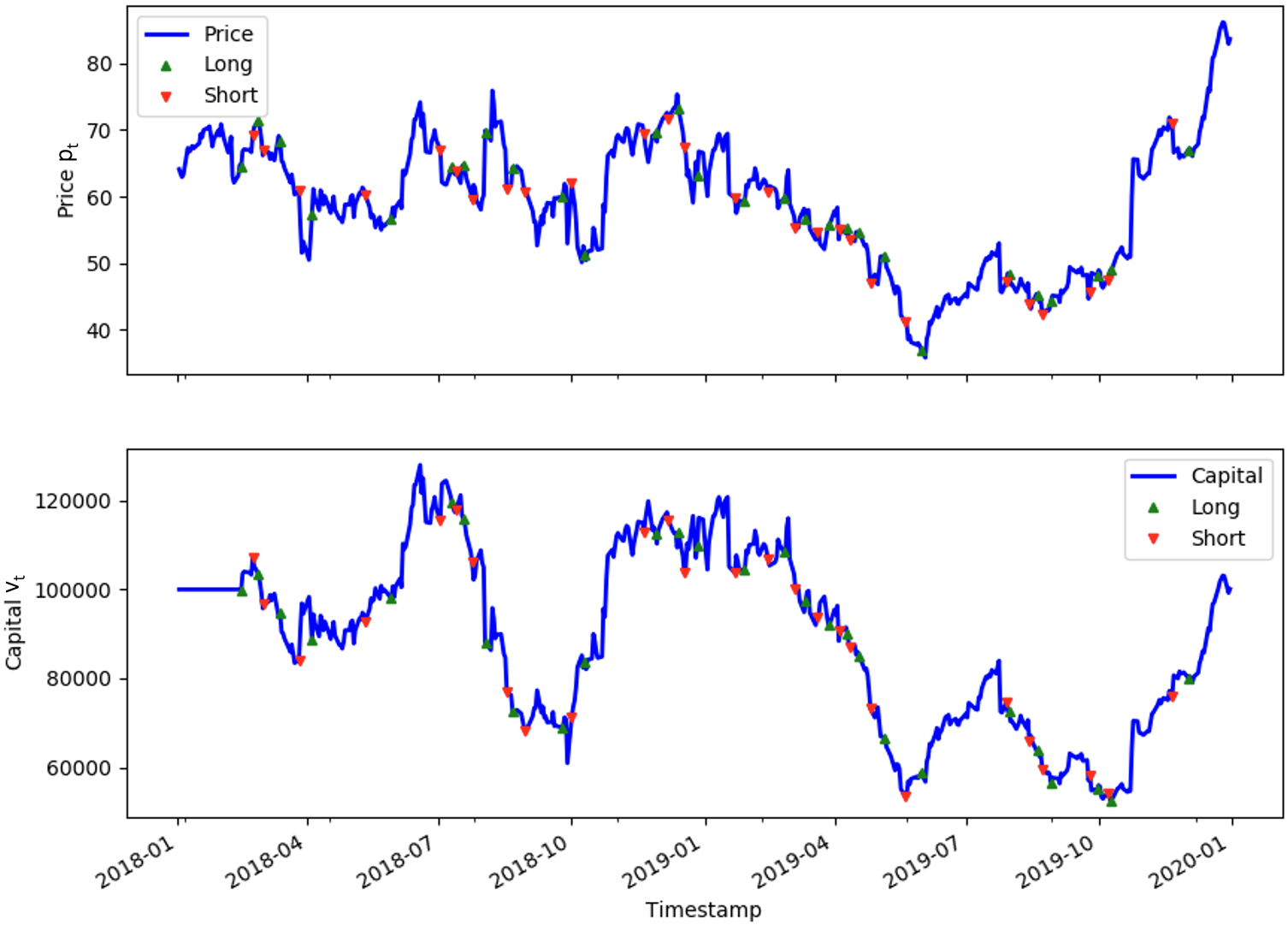}
    \caption{TDQN algorithm execution for the Tesla stock (test set)}
    \label{TestTesla}
\end{figure}

\textbf{Expected performance:} Figure \ref{ExpectedPerformanceTesla} plots the expected performance of the TDQN algorithm for both the training and test sets as a function of the number of training episodes (over 50 iterations). It can be directly noticed that this expected performance is significantly better than the performance achieved by the typical run previously analysed, which can therefore be considered as not really representative of the average behaviour. This highlights a key limitation of the TDQN algorithm: the substantial variance which may result in selecting poorly performing policies compared to the expected performance. The significantly higher performance achieved on the training set also suggests that the DRL algorithm is subject to overfitting in this specific case, despite the multiple regularisation techniques implemented. This overfitting phenomenon can be partially explained by the observation space $\mathcal{O}$ which is too limited to efficiently apprehend the Tesla stock. Even though this overfitting phenomenon does not seem to be too harmful in this particular case, it may lead to poor performance for other stocks.

\begin{figure}
    \centering
    \includegraphics[scale=0.6, trim={0.3cm 0.2cm 1cm 1.15cm}, clip]{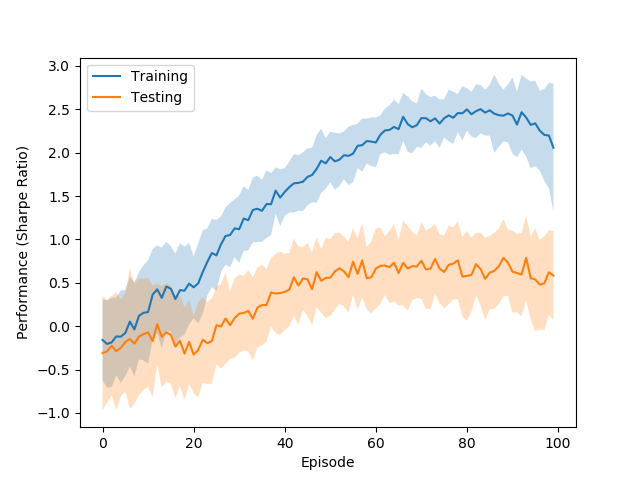}
    \caption{TDQN algorithm expected performance for the Tesla stock}
    \label{ExpectedPerformanceTesla}
\end{figure}

\subsection{Global results - Testbench}
\label{SectionGlobalResults}

As previously suggested in this research paper, the TDQN algorithm is evaluated on the testbench introduced in Section \ref{SectionTestbench}, in order to draw more robust and trustful conclusions. Table \ref{PerformanceAssessment} presents the expected Sharpe ratio achieved by both the TDQN and benchmark trading strategies on the entire set of stocks included in this testbench.\\

\begin{table*}
  \caption{Performance assessment for the entire testbench}
  \label{PerformanceAssessment}
  \centering
  \begin{tabular}{lccccc}
    \toprule
    \multicolumn{1}{c}{\multirow{2}{*}{\textbf{Stock}}} & \multicolumn{5}{c}{\textbf{Sharpe Ratio}}\\
    \cmidrule(r){2-6}
    & \textit{B\&H} & \textit{S\&H} & \textit{TF} & \textit{MR} & \textit{TDQN} \\
    \midrule
    Dow Jones (DIA) & 0.684 & -0.636 & -0.325 & -0.214 & 0.684 \\
    S\&P 500 (SPY) & 0.834 & -0.833 & -0.309 & -0.376 & 0.834 \\
    NASDAQ 100 (QQQ) & 0.845 & -0.806 & 0.264 & 0.060 & 0.845 \\
    FTSE 100 (EZU) & 0.088 & 0.026 & -0.404 & -0.030 & 0.103 \\
    Nikkei 225 (EWJ) & 0.128 & -0.025 & -1.649 & 0.418 & 0.019 \\
    Google (GOOGL) & 0.570 & -0.370 & 0.125 & 0.555 & 0.227 \\
    Apple (AAPL) & 1.239 & -1.593 & 1.178 & -0.609 & 1.424 \\
    Facebook (FB) & 0.371 & -0.078 & 0.248 & -0.168 & 0.151 \\
    Amazon (AMZN) & 0.559 & -0.187 & 0.161 & -1.193 & 0.419 \\
    Microsoft (MSFT) & 1.364 & -1.390 & -0.041 & -0.416 & 0.987 \\
    Twitter (TWTR) & 0.189 & 0.314 & -0.271 & -0.422 & 0.238 \\
    Nokia (NOK) & -0.408 & 0.565 & 1.088 & 1.314 & -0.094 \\
    Philips (PHIA.AS) & 1.062 & -0.672 & -0.167 & -0.599 & 0.675 \\
    Siemens (SIE.DE) & 0.399 & -0.265 & 0.525 & 0.526 & 0.426 \\
    Baidu (BIDU) & -0.699 & 0.866 & -1.209 & 0.167 & 0.080 \\
    Alibaba (BABA) & 0.357 & -0.139 & -0.068 & 0.293 & 0.021 \\
    Tencent (0700.HK) & -0.013 & 0.309 & 0.179 & -0.466 & -0.198 \\
    Sony (6758.T) & 0.794 & -0.655 & -0.352 & 0.415 & 0.424 \\
    JPMorgan Chase (JPM) & 0.713 & -0.743 & -1.325 & -0.004 & 0.722 \\
    HSBC (HSBC) & -0.518 & 0.725 & -1.061 & 0.447 & 0.011 \\
    CCB (0939.HK) & 0.026 & 0.165 & -1.163 & -0.388 & 0.202 \\
    ExxonMobil (XOM) & 0.055 & 0.132 & -0.386 & -0.673 & 0.098 \\
    Shell (RDSA.AS) & 0.488 & -0.238 & -0.043 & 0.742 & 0.425 \\
    PetroChina (PTR) & -0.376 & 0.514 & -0.821 & -0.238 & 0.156 \\
    Tesla (TSLA) & 0.508 & -0.154 & -0.987 & 0.358 & 0.621 \\
    Volkswagen (VOW3.DE) & 0.384 & -0.208 & -0.361 & 0.601 & 0.216 \\
    Toyota (7203.T) & 0.352 & -0.242 & -1.108 & -0.378 & 0.304 \\
    Coca Cola (KO) & 1.031 & -0.871 & -0.236 & -0.394 & 1.068 \\
    AB InBev (ABI.BR) & -0.058 & 0.275 & 0.036 & -1.313 & 0.187 \\
    Kirin (2503.T) & 0.106 & 0.156 & -1.441 & 0.313 & 0.852 \\
    \midrule
    \midrule
    \multicolumn{1}{c}{\textbf{Average}} & 0.369 & -0.202 & -0.331 & -0.056 & 0.404 \\
    \bottomrule
  \end{tabular}
\end{table*}

Regarding the performance achieved by the benchmark trading strategies, it is important to differentiate the passive strategies (B\&H and S\&H) from the active ones (TF and MR). Indeed, this second family of trading strategies has more potential at the cost of an extra non-negligible risk: continuous speculation. Because the stock markets were mostly \textit{bullish} (price $p_t$ mainly increasing over time) with some instabilities during the test set trading period, it is not surprising to see the buy and hold strategy outperforming the other benchmark trading strategies. In fact, neither the trend following nor the mean reversion strategy managed to generate satisfying results on average on this testbench. It clearly indicates that there is a major difficulty to actively trade in such market conditions. This poorer performance can also be explained by the fact that such strategies are generally well suited to exploit specific financial patterns, but they lack versatility and thus often fail to achieve good average performance on a large set of stocks presenting diverse characteristics. Moreover, such strategies are generally more impacted by the trading costs due their higher trading frequency (for relatively short moving averages durations, as it is the case in this research paper).\\

Concerning the innovative trading strategy, the TDQN algorithm achieves promising results on the testbench, outperforming the benchmark active trading strategies on average. Nevertheless, the DRL trading strategy only barely surpasses the buy and hold strategy on these particular bullish markets which are so favourable to this simple passive strategy. Interestingly, it should be noted that the performance of the TDQN algorithm is identical or very close to the performance of the passive trading strategies (B\&H and S\&H) for multiple stocks. This is explained by the fact that the DRL strategy efficiently learns to tend toward a passive trading strategy when the uncertainty associated to active trading increases. It should also be emphasized that the TDQN algorithm is neither a trend following nor a mean reversion trading strategy as both financial patterns can be efficiently handled in practice. Thus, the main advantage of the DRL trading strategy is certainly its versatility and its ability to efficiently handle various markets presenting diverse characteristics.

\subsection{Discount factor discussion}
\label{SectionDiscountFactor}

As previously explained in Section \ref{SectionRLFormalism}, the discount factor $\gamma$ is concerned with the importance of future rewards. In the scope of this algorithmic trading problem, the proper tuning of this parameter is not trivial due to the significant uncertainty of the future. On the one hand, the desired trading policy should be long-term oriented ($\gamma \rightarrow 1$), in order to avoid a too high trading frequency and being exposed to considerable trading costs. On the other hand, it would be unwise to place too much importance on a stock market future which is particularly uncertain ($\gamma \rightarrow 0$). Therefore, a trade-off intuitively exists for the discount factor parameter.\\

This reasoning is validated by the multiple experiments performed to tune the parameter $\gamma$. Indeed, it was observed that there is an optimal value for the discount factor, which is neither too small nor too large. Additionally, these experiments highlighted the hidden link between the discount factor and the trading frequency, due to the trading costs. From the point of view of the RL agent, these costs represent an obstacle to overcome for a change in trading position to occur, due to the immediate reduced (and often negative) reward received. It models the fact that the trading agent should be sufficiently confident about the future in order to overcome the extra risk associated with the trading costs. The discount factor determining the importance assigned to the future, a small value for the parameter $\gamma$ will inevitably reduce the tendency of the RL agent to change its trading position, which decreases the trading frequency of the TDQN algorithm.

\subsection{Trading costs discussion}
\label{SectionTradingCostsDiscussion}

The analysis of the trading costs influence on a trading strategy behaviour and performance is capital, due to the fact that such costs represent an extra risk to mitigate. A major motivation for studying DRL solutions rather than pure prediction techniques that could also be based on DL architectures is related to the trading costs. As previously explained in Section \ref{SectionFormalisation}, the RL formalism enables the consideration of these additional costs directly into the decision-making process. The optimal policy is learned according to the trading costs value. On the contrary, a purely predictive approach would only output predictions about the future market direction or prices without any indications regarding an appropriate trading strategy taking into account the trading costs. Although this last approach offers more flexibility and could certainly lead to well-performing trading strategies, it is less efficient by design.\\

In order to illustrate the ability of the TDQN algorithm to automatically and efficiently adapt to different trading costs, Figure \ref{TradingCostsAnalysis} presents the behaviour of the DRL trading strategy for three different costs values, all other parameters remaining unchanged. It can clearly be observed that the TDQN algorithm effectively reduces its trading frequency when the trading costs increase, as expected. When these costs become too high, the DRL algorithm simply stops actively trading and adopts a passive approach (buy and hold or sell and hold strategies).

\begin{figure*}
    \centering
    \begin{subfigure}{0.333\textwidth}
        \includegraphics[scale=0.275, trim={0.8cm 2cm 0.8cm 2cm}, clip]{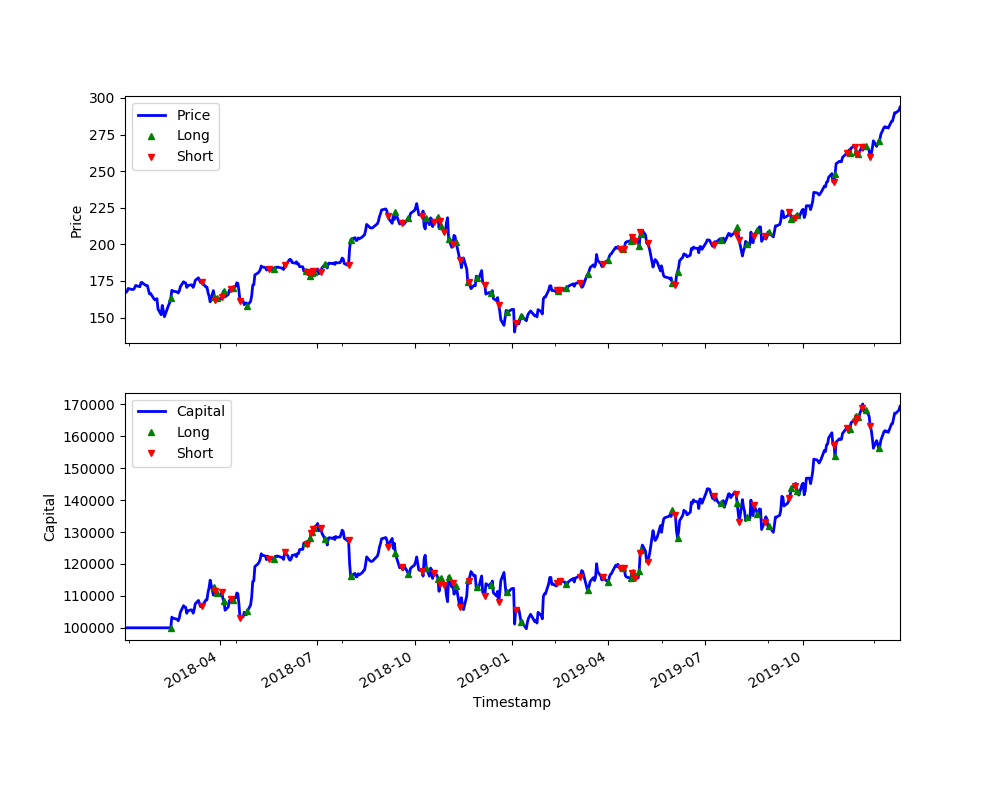}
        \caption{Trading costs: 0\%}
    \end{subfigure}%
    \begin{subfigure}{0.333\textwidth}
        \includegraphics[scale=0.275, trim={0.8cm 2cm 0.8cm 2cm}, clip]{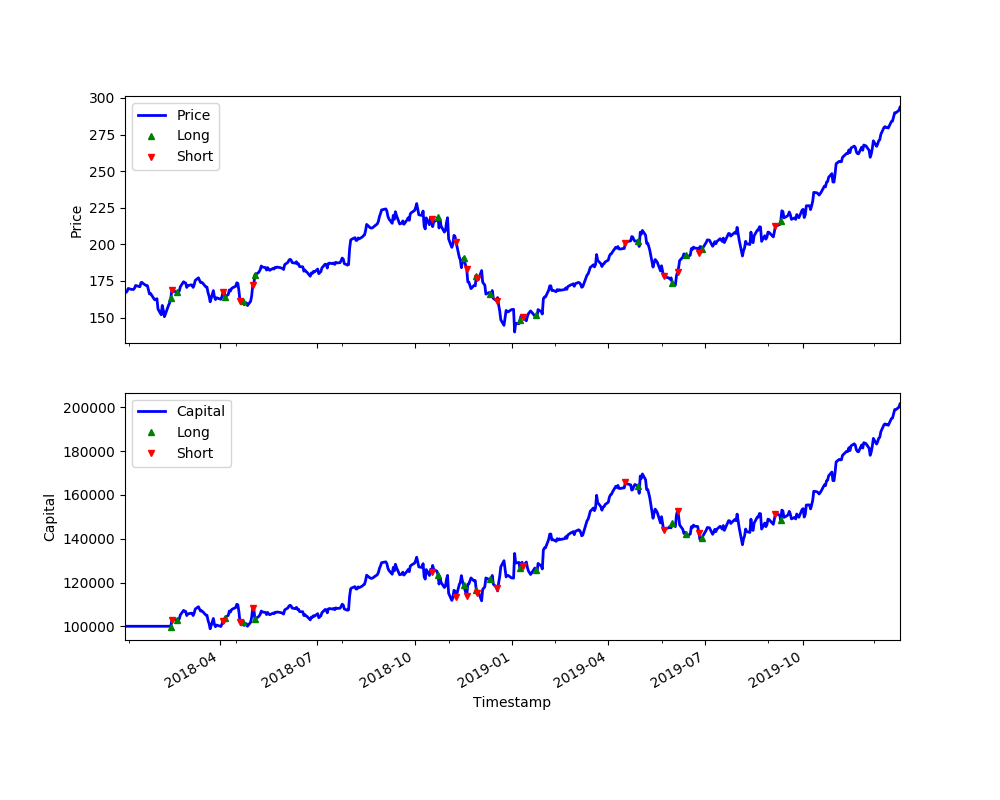}
        \caption{Trading costs: 0.1\%}
    \end{subfigure}%
    \begin{subfigure}{0.333\textwidth}
        \includegraphics[scale=0.275, trim={0.8cm 2cm 0.8cm 2cm}, clip]{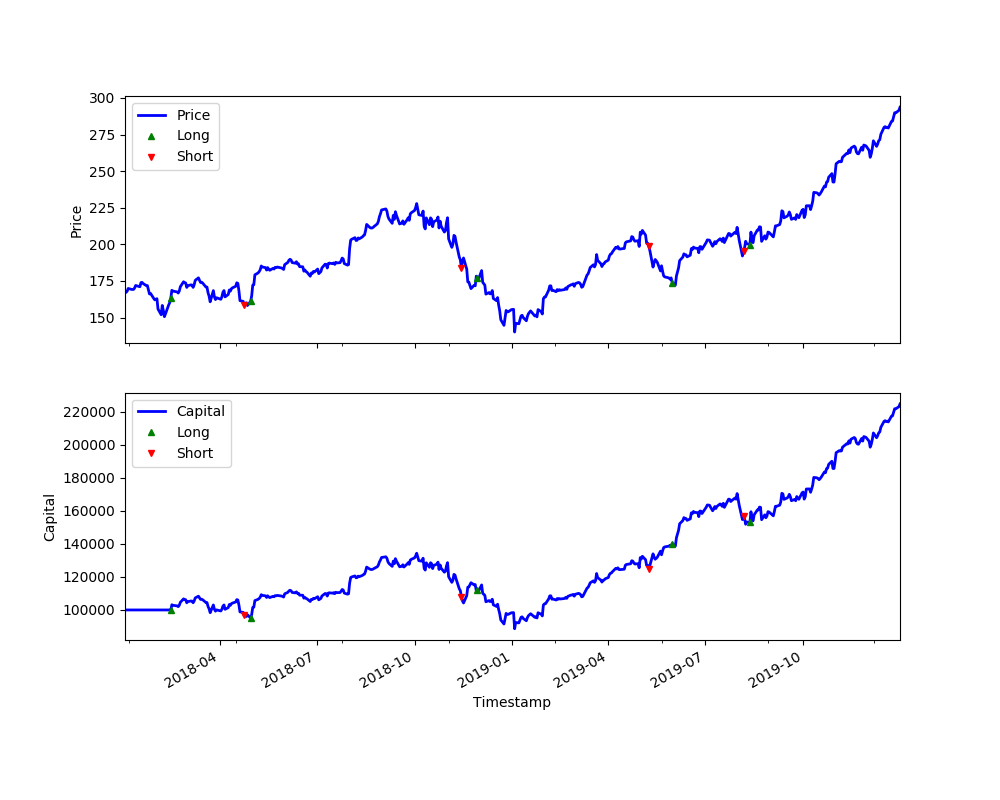}
        \caption{Trading costs: 0.2\%}
    \end{subfigure}
    \caption{Impact of the trading costs on the TDQN algorithm, for the Apple stock}
    \label{TradingCostsAnalysis}
\end{figure*}

\subsection{Core challenges}
\label{SectionChallenges}

Nowadays, the main DRL solutions successfully applied to real-life problems concern specific environments with particular properties such as games (see e.g. the famous AlphaGo algorithm developed by Google Deepmind \cite{Silver2016}). In this research paper, an entirely different environment characterised by a significant complexity and a considerable uncertainty is studied with the algorithmic trading problem. Obviously, multiple challenges were faced during the research around the TDQN algorithm, the major ones being summarised hereafter.\\

Firstly, the extremely poor observability of the trading environment is a characteristic that significantly limits the performance of the TDQN algorithm. Indeed, the amount of information at the disposal of the RL agent is really not sufficient to accurately explain the financial phenomena occurring during training, which is necessary to efficiently learn to trade. Secondly, although the distribution of the daily returns is continuously changing, the past is required to be representative enough of the future for the TDQN algorithm to achieve good results. This makes the DRL trading strategy particularly sensitive to significant market regime shifts. Thirdly, the TDQN algorithm overfitting tendency has to be properly handled in order to obtain a reliable trading strategy. As suggested in \cite{Zhang2018}, more rigorous evaluation protocols are required in RL due to the strong tendency of common DRL techniques to overfit. More research on this particular topic is required for DRL techniques to fit a broader range of real-life applications. Lastly, the substantial variance of DRL algorithms such as DQN makes it rather difficult to successfully apply these algorithms to certain problems, especially when the training and test sets differ considerably. This is a key limitation of the TDQN algorithm which was previously highlighted for the Tesla stock.

\section{Conclusion}
\label{SectionConclusion}

This scientific research paper presents the Trading Deep Q-Network algorithm (TDQN), a deep reinforcement learning (DRL) solution to the algorithmic trading problem of determining the optimal trading position at any point in time during a trading activity in stock markets. Following a rigorous performance assessment, this innovative trading strategy achieves promising results, surpassing on average the benchmark trading strategies. Moreover, the TDQN algorithm demonstrates multiple benefits compared to more classical approaches, such as an appreciable versatility and a remarkable robustness to diverse trading costs. Additionally, such data-driven approach presents the major advantage of suppressing the complex task of defining explicit rules suited to the particular financial markets considered.\\ 

Nevertheless, the performance of the TDQN algorithm could still be improved, from both a generalisation and a reproducibility point of view, to cite a few. Several research directions are suggested to upgrade the DRL solution, such as the use of LSTM layers into the deep neural network which should help to better process the financial time-series data, see e.g. \cite{Hausknecht2015}. Another example is the consideration of the numerous improvements implemented in the Rainbow algorithm, which are detailed in \cite{Sutton2018}, \cite{Hasselt2016}, \cite{Wang2016}, \cite{Schaul2016}, \cite{Bellemare2017}, \cite{Fortunato2018} and \cite{Hessel2018}. Another interesting research direction is the comparison of the TDQN algorithm with Policy Optimisation DRL algorithms such as the Proximal Policy Optimisation (PPO - \cite{Schulman2017}) algorithm.\\

The last major research direction suggested concerns the formalisation of the algorithmic trading problem into a reinforcement learning one. Firstly, the observation space $\mathcal{O}$ should be extended to enhance the observability of the trading environment. Similarly, some constraints about the action space $\mathcal{A}$ could be relaxed in order to enable new trading possibilities. Secondly, advanced RL reward engineering should be performed to narrow the gap between the RL objective and the Sharpe ratio maximisation objective. Finally, an interesting and promising research direction is the consideration of distributions instead of expected values in the TDQN algorithm in order to encompass the notion of risk and to better handle uncertainty.

\section*{Acknowledgements}

Thibaut Th\'{e}ate is a Research Fellow of the F.R.S.-FNRS, of which he acknowledges the financial support.

\bibliography{manuscript}

\begin{thebibliography}{}

\bibitem[Ar{\'e}valo et~al., 2016]{Arevalo2016}
Ar{\'e}valo, A., Ni{\~n}o, J., Hern{\'a}ndez, G., and Sandoval, J. (2016).
\newblock High-{F}requency {T}rading {S}trategy {B}ased on {D}eep {N}eural
  {N}etworks.
\newblock {\em ICIC}.

\bibitem[Arulkumaran et~al., 2017]{Arulkumaran2017}
Arulkumaran, K., Deisenroth, M.~P., Brundage, M., and Bharath, A.~A. (2017).
\newblock A {B}rief {S}urvey of {D}eep {R}einforcement {L}earning.
\newblock {\em CoRR}, abs/1708.05866.

\bibitem[Bailey et~al., 2014]{Bailey2014}
Bailey, D.~H., Borwein, J.~M., de~Prado, M.~L., and Zhu, Q.~J. (2014).
\newblock Pseudo-{M}athematics and {F}inancial {C}harlatanism: {T}he {E}ffects
  of {B}acktest {O}verfitting on {O}ut-of-{S}ample {P}erformance.
\newblock {\em Notice of the American Mathematical Society}, pages 458--471.

\bibitem[Bao et~al., 2017]{Bao2017}
Bao, W.~N., Yue, J., and Rao, Y. (2017).
\newblock A {D}eep {L}earning {F}ramework for {F}inancial {T}ime {S}eries using
  {S}tacked {A}utoencoders and {L}ong-{S}hort {T}erm {M}emory.
\newblock {\em PloS one}, 12.

\bibitem[Bellemare et~al., 2017]{Bellemare2017}
Bellemare, M.~G., Dabney, W., and Munos, R. (2017).
\newblock A {D}istributional {P}erspective on {R}einforcement {L}earning.
\newblock {\em CoRR}, abs/1707.06887.

\bibitem[Bollen et~al., 2011]{Bollen2011}
Bollen, J., Mao, H., and jun Zeng, X. (2011).
\newblock Twitter {M}ood {P}redicts the {S}tock {M}arket.
\newblock {\em J. Comput. Science}, 2:1--8.

\bibitem[Boukas et~al., 2020]{Boukas2020}
Boukas, I., Ernst, D., Th{\'e}ate, T., Bolland, A., Huynen, A., Buchwald, M.,
  Wynants, C., and Corn{\'e}lusse, B. (2020).
\newblock A {D}eep {R}einforcement {L}earning {F}ramework for {C}ontinuous
  {I}ntraday {M}arket {B}idding.
\newblock {\em ArXiv}, abs/2004.05940.

\bibitem[Busoniu et~al., 2010]{Busoniu2010}
Busoniu, L., Babuska, R., De~Schutter, B., and Ernst, D. (2010).
\newblock {\em Reinforcement Learning and Dynamic Programming using Function
  Approximators}.
\newblock CRC Press.

\bibitem[Carapu\c{c}o et~al., 2018]{Carapuco2018}
Carapu\c{c}o, J., Neves, R.~F., and Horta, N. (2018).
\newblock Reinforcement {L}earning applied to {F}orex {T}rading.
\newblock {\em Appl. Soft Comput.}, 73:783--794.

\bibitem[Chan, 2009]{Chan2009}
Chan, E.~P. (2009).
\newblock {\em Quantitative Trading: How to Build Your Own Algorithmic Trading
  Business}.
\newblock Wiley.

\bibitem[Chan, 2013]{Chan2013}
Chan, E.~P. (2013).
\newblock {\em Algorithmic Trading: Winning Strategies and Their Rationale}.
\newblock Wiley.

\bibitem[Dempster and Leemans, 2006]{Dempster2006}
Dempster, M. A.~H. and Leemans, V. (2006).
\newblock An {A}utomated {FX} {T}rading {S}ystem using {A}daptive
  {R}einforcement {L}earning.
\newblock {\em Expert Syst. Appl.}, 30:543--552.

\bibitem[Deng et~al., 2017]{Deng2017}
Deng, Y., Bao, F., Kong, Y., Ren, Z., and Dai, Q. (2017).
\newblock Deep {D}irect {R}einforcement {L}earning for {F}inancial {S}ignal
  {R}epresentation and {T}rading.
\newblock {\em IEEE Transactions on Neural Networks and Learning Systems},
  28:653--664.

\bibitem[Fortunato et~al., 2018]{Fortunato2018}
Fortunato, M., Azar, M.~G., Piot, B., Menick, J., Hessel, M., Osband, I.,
  Graves, A., Mnih, V., Munos, R., Hassabis, D., Pietquin, O., Blundell, C.,
  and Legg, S. (2018).
\newblock Noisy {N}etworks for {E}xploration.
\newblock {\em CoRR}, abs/1706.10295.

\bibitem[Goodfellow et~al., 2016]{Goodfellow2016}
Goodfellow, I., Bengio, Y., and Courville, A. (2016).
\newblock {\em Deep Learning}.
\newblock MIT Press.

\bibitem[Goodfellow et~al., 2015]{Goodfellow2015}
Goodfellow, I.~J., Bengio, Y., and Courville, A.~C. (2015).
\newblock Deep {L}earning.
\newblock {\em Nature}, 521:436--444.

\bibitem[Hausknecht and Stone, 2015]{Hausknecht2015}
Hausknecht, M.~J. and Stone, P. (2015).
\newblock Deep {R}ecurrent {Q}-{L}earning for {P}artially {O}bservable {MDPs}.
\newblock {\em CoRR}, abs/1507.06527.

\bibitem[Hendershott et~al., 2011]{Hendershott2011}
Hendershott, T., Jones, C.~M., and Menkveld, A.~J. (2011).
\newblock Does {A}lgorithmic {T}rading {I}mprove {L}iquidity?
\newblock {\em Journal of Finance}, 66:1--33.

\bibitem[Hessel et~al., 2017]{Hessel2018}
Hessel, M., Modayil, J., van Hasselt, H.~P., Schaul, T., Ostrovski, G., Dabney,
  W., Horgan, D., Piot, B., Azar, M.~G., and Silver, D. (2017).
\newblock Rainbow: {C}ombining {I}mprovements in {D}eep {R}einforcement
  {L}earning.
\newblock {\em CoRR}, abs/1710.02298.

\bibitem[Ioannidis, 2005]{Ioannidis2005}
Ioannidis, J. P.~A. (2005).
\newblock Why {M}ost {P}ublished {R}esearch {F}indings {A}re {F}alse.
\newblock {\em PLoS Med}, 2:124.

\bibitem[Ioffe and Szegedy, 2015]{Ioffe2015}
Ioffe, S. and Szegedy, C. (2015).
\newblock {B}atch {N}ormalization: {A}ccelerating {D}eep {N}etwork {T}raining
  by {R}educing {I}nternal {C}ovariate {S}hift.
\newblock {\em CoRR}, abs/1502.03167.

\bibitem[Kingma and Ba, 2015]{Kingma2015}
Kingma, D.~P. and Ba, J. (2015).
\newblock Adam: {A} {M}ethod for {S}tochastic {O}ptimization.
\newblock {\em CoRR}, abs/1412.6980.

\bibitem[LeCun et~al., 2015]{LeCun2015}
LeCun, Y., Bengio, Y., and Hinton, G. (2015).
\newblock Deep {L}earning.
\newblock {\em Nature}, 521.

\bibitem[Leinweber and Sisk, 2011]{Leinweber2011}
Leinweber, D. and Sisk, J. (2011).
\newblock Event-{D}riven {T}rading and the “{N}ew {N}ews”.
\newblock {\em The Journal of Portfolio Management}, 38:110--124.

\bibitem[Li, 2017]{Li2017}
Li, Y. (2017).
\newblock Deep {R}einforcement {L}earning: An {O}verview.
\newblock {\em CoRR}, abs/1701.07274.

\bibitem[Mnih et~al., 2013]{Mnih2013}
Mnih, V., Kavukcuoglu, K., Silver, D., Graves, A., Antonoglou, I., Wierstra,
  D., and Riedmiller, M.~A. (2013).
\newblock Playing {A}tari with {D}eep {R}einforcement {L}earning.
\newblock {\em CoRR}, abs/1312.5602.

\bibitem[Mnih et~al., 2015]{Mnih2015}
Mnih, V., Kavukcuoglu, K., Silver, D., Rusu, A.~A., Veness, J., Bellemare,
  M.~G., Graves, A., Riedmiller, M.~A., Fidjeland, A., Ostrovski, G., Petersen,
  S., Beattie, C., Sadik, A., Antonoglou, I., King, H., Kumaran, D., Wierstra,
  D., Legg, S., and Hassabis, D. (2015).
\newblock Human-{L}evel {C}ontrol through {D}eep {R}einforcement {L}earning.
\newblock {\em Nature}, 518:529--533.

\bibitem[Moody and Saffell, 2001]{Moody2001}
Moody, J.~E. and Saffell, M. (2001).
\newblock Learning to {T}rade via {D}irect {R}einforcement.
\newblock {\em IEEE transactions on neural networks}, 12 4:875--89.

\bibitem[Narang, 2009]{Narang2009}
Narang, R.~K. (2009).
\newblock {\em Inside the Black Box}.
\newblock Wiley.

\bibitem[Nuij et~al., 2014]{Nuij2014}
Nuij, W., Milea, V., Hogenboom, F., Frasincar, F., and Kaymak, U. (2014).
\newblock An {A}utomated {F}ramework for {I}ncorporating {N}ews into {S}tock
  {T}rading {S}trategies.
\newblock {\em IEEE Transactions on Knowledge and Data Engineering},
  26:823--835.

\bibitem[Nuti et~al., 2011]{Nuti2011}
Nuti, G., Mirghaemi, M., Treleaven, P.~C., and Yingsaeree, C. (2011).
\newblock Algorithmic {T}rading.
\newblock {\em Computer}, 44:61--69.

\bibitem[Schaul et~al., 2016]{Schaul2016}
Schaul, T., Quan, J., Antonoglou, I., and Silver, D. (2016).
\newblock Prioritized {E}xperience {R}eplay.
\newblock {\em CoRR}, abs/1511.05952.

\bibitem[Schulman et~al., 2017]{Schulman2017}
Schulman, J., Wolski, F., Dhariwal, P., Radford, A., and Klimov, O. (2017).
\newblock Proximal {P}olicy {O}ptimization {A}lgorithms.
\newblock {\em CoRR}, abs/1707.06347.

\bibitem[Shao et~al., 2019]{Shao2019}
Shao, K., Tang, Z., Zhu, Y., Li, N., and Zhao, D. (2019).
\newblock A {S}urvey of {D}eep {R}einforcement {L}earning in {V}ideo {G}ames.
\newblock {\em ArXiv}, abs/1912.10944.

\bibitem[Silver et~al., 2016]{Silver2016}
Silver, D., Huang, A., Maddison, C.~J., Guez, A., Sifre, L., van~den Driessche,
  G., Schrittwieser, J., Antonoglou, I., Panneershelvam, V., Lanctot, M.,
  Dieleman, S., Grewe, D., Nham, J., Kalchbrenner, N., Sutskever, I.,
  Lillicrap, T.~P., Leach, M., Kavukcuoglu, K., Graepel, T., and Hassabis, D.
  (2016).
\newblock Mastering the {G}ame of {G}o with {D}eep {N}eural {N}etworks and
  {T}ree {S}earch.
\newblock {\em Nature}, 529:484--489.

\bibitem[Sutton and Barto, 2018]{Sutton2018}
Sutton, R.~S. and Barto, A.~G. (2018).
\newblock {\em Reinforcement Learning: An Introduction}.
\newblock The MIT Press, second edition.

\bibitem[Szepesvari, 2010]{Szepesvri2010}
Szepesvari, C. (2010).
\newblock {\em Algorithms for Reinforcement Learning}.
\newblock Morgan and Claypool Publishers.

\bibitem[Treleaven et~al., 2013]{Treleaven2013}
Treleaven, P.~C., Galas, M., and Lalchand, V. (2013).
\newblock Algorithmic {T}rading {R}eview.
\newblock {\em Commun. ACM}, 56:76--85.

\bibitem[van Hasselt et~al., 2015]{Hasselt2016}
van Hasselt, H.~P., Guez, A., and Silver, D. (2015).
\newblock Deep {R}einforcement {L}earning with {D}ouble {Q}-{L}earning.
\newblock {\em CoRR}, abs/1509.06461.

\bibitem[Wang et~al., 2015]{Wang2016}
Wang, Z., de~Freitas, N., and Lanctot, M. (2015).
\newblock Dueling {N}etwork {A}rchitectures for {D}eep {R}einforcement
  {L}earning.
\newblock {\em CoRR}, abs/1511.06581.

\bibitem[Watkins and Dayan, 1992]{Watkins1992}
Watkins, C. J. C.~H. and Dayan, P. (1992).
\newblock Technical {N}ote: {Q}-{L}earning.
\newblock {\em Machine Learning}, 8:279--292.

\bibitem[Zhang et~al., 2018]{Zhang2018}
Zhang, C., Vinyals, O., Munos, R., and Bengio, S. (2018).
\newblock A {S}tudy on {O}verfitting in {D}eep {R}einforcement {L}earning.
\newblock {\em CoRR}, abs/1804.06893.

\end{thebibliography}

\clearpage

\appendix

\section{Derivation of action space \texorpdfstring{$\mathcal{A}$}{}}
\label{AppendixA}

\begin{theorem} The RL action space $\mathcal{A}$ admits an upper bound $\overline{Q_t}$ such that:
$$\overline{Q_t} = \frac{v^c_t}{p_t\ (1+C)}$$
\end{theorem}

\begin{proof} The upper bound of the RL action space $\mathcal{A}$ is derived from the fact that the cash value $v^c_t$ has to remain positive over the entire trading horizon (Equation \ref{EquationConstraint1}). Making the hypothesis that $v^c_t \ge 0$, the number of shares $Q_t$ traded by the RL agent at time step $t$ has to be set such that $v^c_{t+1} \ge 0$ as well. Introducing this condition into Equation \ref{EquationCashValueCorrected} expressing the update of the cash value, the following expression is obtained:

$$v^c_t - Q_t\ p_t - C\ |Q_{t}|\ p_t \ge 0$$

\noindent Two cases arise depending on the value of $Q_t$:\\

\noindent \underline{Case of $Q_t < 0$}: The previous expression becomes\\
$v^c_t - Q_t\ p_t + C\ Q_{t}\ p_t \ge 0$.\\
$\Leftrightarrow Q_t \le \frac{v^c_t}{p_t\ (1-C)}$.\\
The expression on the right side of the inequality is always positive due to the hypothesis that $v^c_t \ge 0$. Because $Q_t$ is negative in this case, the condition is always satisfied.\\
    
\noindent \underline{Case of $Q_t \ge 0$}: The previous expression becomes\\
$v^c_t - Q_t\ p_t - C\ Q_{t}\ p_t \ge 0$.\\
$\Leftrightarrow Q_t \le \frac{v^c_t}{p_t\ (1+C)}$.\\
This condition represents the upper bound (positive) of the RL action space $\mathcal{A}$.
\end{proof}

\vspace{0.3cm}

\begin{theorem} The RL action space $\mathcal{A}$ admits a lower bound $\underline{Q_t}$ such that:

$$\underline{Q_t} = \left\{\begin{matrix}
                        \frac{\Delta_t}{p_t\ \epsilon(1 + C)}\ \ \ \ \ \ \ \ \ \text{if } \Delta_t \ge 0\\ 
                        \frac{\Delta_t}{p_t\ (2C + \epsilon(1 + C))} \ \ \ \text{if } \Delta_t < 0
                    \end{matrix}
                    \right.$$

\noindent with $\Delta_t = -v^c_t-n_t\ p_t\ (1+\epsilon)(1+C)$.
\end{theorem}

\begin{proof}
The lower bound of the RL action space $\mathcal{A}$ is derived from the fact that the cash value $v^c_t$ has to be sufficient to get back to a neutral position ($n_t = 0$) over the entire trading horizon (Equation \ref{EquationConstraint2Bis}). Making the hypothesis that this condition is satisfied at time step $t$, the number of shares $Q_t$ traded by the RL agent should be such that this condition remains true at the next time step $t+1$. Introducing this constraint into Equation \ref{EquationCashValueCorrected}, the following inequality is obtained:

$$v^c_t - Q_t\ p_t - C\ |Q_{t}|\ p_t \ge -(n_t + Q_t)\ p_t\ (1+C)(1+\epsilon)$$

\noindent Two cases arise depending on the value of $Q_t$:\\

\noindent \underline{Case of $Q_t \ge 0$}: The previous expression becomes\\
$v^c_t - Q_t\ p_t - C\ Q_{t}\ p_t \ge -(n_t + Q_t)\ p_t\ (1+C)(1+\epsilon)$\\
$\Leftrightarrow v^c_t \ge - n_t\ p_t\ (1+C)(1+\epsilon) - Q_t\ p_t\ \epsilon\ (1+C)$\\
$\Leftrightarrow Q_t \ge \frac{-v^c_t - n_t\ p_t\ (1+C)(1+\epsilon)}{p_t\ \epsilon\ (1+C)}$\\
The expression on the right side of the inequality represents the first lower bound for the RL action space $\mathcal{A}$.\\

\noindent \underline{Case of $Q_t < 0$}: The previous expression becomes\\
$v^c_t - Q_t\ p_t + C\ Q_{t}\ p_t \ge -(n_t + Q_t)\ p_t\ (1+C)(1+\epsilon)$\\
$\Leftrightarrow v^c_t \ge - n_t\ p_t\ (1+C)(1+\epsilon) - Q_t\ p_t\ (2C + \epsilon + \epsilon C)$\\
$\Leftrightarrow Q_t \ge \frac{-v^c_t - n_t\ p_t\ (1+C)(1+\epsilon)}{p_t\ (2C + \epsilon(1 + C))}$\\
The expression on the right side of the inequality represents the second lower bound for the RL action space $\mathcal{A}$.\\

Both lower bounds previously derived have the same numerator, which is denoted $\Delta_t$ from now on. This quantity represents the difference between the maximum assumed cost to get back to a neutral position at the next time step $t+1$ and the current cash value of the agent $v^c_t$. The expression tests whether the agent can pay its debt in the worst assumed case or not at the next time step, if nothing is done at the current time step ($Q_t = 0$). Two cases arise depending on the sign of the quantity $\Delta_t$:\\

\noindent \underline{Case of $\Delta_t < 0$}: The trading agent has no problem paying its debt in the situation previously described. This is always true when the agent owns a positive number of shares ($n_t \ge 0$). This is also always true when the agent owns a negative number of shares ($n_t < 0$) and when the price decreases ($p_{t} < p_{t-1}$) due to the hypothesis that Equation \ref{EquationConstraint2Bis} was verified for time step $t$. In this case, the most constraining lower bound of the two is the following:

$$\underline{Q_t} = \frac{\Delta_t}{p_t\ (2C + \epsilon(1 + C))}$$

\vspace{0.2cm}

\noindent \underline{Case of $\Delta_t \ge 0$}: The trading agent may have problem paying its debt in the situation previously described. Following a similar reasoning than for the previous case, the most constraining lower bound of the two is the following:

$$\underline{Q_t} = \frac{\Delta_t}{p_t\ \epsilon(1 + C)}$$

\end{proof}

\end{document}